\documentclass[10pt,journal,epsfig]{IEEEtran}

\usepackage[dvips]{graphicx}
\usepackage{graphicx}
\usepackage{amssymb}
\usepackage{cite}
\usepackage{subfigure}
\usepackage{amsmath}

\usepackage{subeqnarray}
\usepackage{cases}
\usepackage[centerlast]{caption2}

\begin{document}

\IEEEoverridecommandlockouts
\title{MIMO Secret Communications Against an Active Eavesdropper% with Multi-Antennas
%On Exploiting Multi-Antennas to Improve Secret Communications Against an Active Eavesdropper
}
\author{ %Authors %1, Author 2, Author 3, Author 4
Lingxiang Li, Athina P. Petropulu, ~\IEEEmembership{Fellow,~IEEE}, Zhi Chen, ~\IEEEmembership{Member,~IEEE}%, Jun Fang, ~\IEEEmembership{Member,~IEEE}
\thanks{
%Copyright (c) 2016 IEEE. Personal use of this material is permitted. However, permission to use this material for any other purposes must be obtained from the IEEE by sending a request to pubs-permissions@ieee.org.
Lingxiang Li and Zhi Chen are with the National Key Laboratory of Science and Technology on Communications,
UESTC, Chengdu 611731, China (e-mails: lingxiang.li@rutgers.edu; chenzhi@uestc.edu.cn). The work was performed
when L. Li was a visiting student at Rutgers University.}
\thanks{Athina P. Petropulu is with the Department of Electrical and Computer Engineering, Rutgers--The State University of New Jersey, New Brunswick, NJ 08854 USA (e-mail: athinap@rci.rutgers.edu).}
%\thanks{This work was supported in part by the National Natural Science Foundation
%of China under Grant 61571089, and by the High-Tech Research and Development (863) Program of China under Grand 2015AA01A707.}
}

%\title{On the Secrecy Capacity of a MIMO Gaussian Wiretap Channel with a Cooperative Jammer}
%\author{ %Author 1, Author 2, Author 3, Author 4
%Lingxiang Li, Zhi Chen, Jun Fang, ~\IEEEmembership{Member,~IEEE},
%and \\ Athina P. Petropulu, ~\IEEEmembership{Fellow,~IEEE}
%\thanks{Lingxiang Li, Zhi Chen, and Jun Fang are with the National Key Laboratory of Science and Technology on Communications,
%University of Electronic Science and Technology of China, Chengdu 610054, China (e-mails:lingxiang\_li\_uestc@hotmail.com; \{chenzhi, JunFang\}@uestc.edu.cn)}
%\thanks{A. P. Petropulu is with the Department of Electrical and Computer Engineering, Rutgers--The State University of New Jersey, New Brunswick, NJ 08854 USA (e-mail: athinap@rci.rutgers.edu).}
%\thanks{This work was supported in part by the Important National Science and
%Technology Specific Projects of China under Grant 2014ZX03004003, and by
%the Sichuan Province Project under Grant 2012FZ0119.}
%}

\maketitle

\begin{abstract}
%This paper considers an
%\emph{Alice}-\emph{Bob} pair whose secret communications are threatened by an active,
%multi-antenna \emph{Eve}, who by appropriately allocating its antennas is capable of jamming as well as eavesdropping in Full-Duplex (FD) mode.
%As countermeasure, we propose to use a
%FD \emph{Bob}, who jams \emph{Eve} while receiving \emph{Alice}'s message. We provide the optimal
%number of receive/transmit antennas at \emph{Bob} with respect to the achievable secrecy degrees of freedom (S.D.o.F.),
%based on which we determine the maximum
%achievable S.D.o.F., as a function of the number of antennas at each terminal and the
%number of receive/transmit antennas at \emph{Eve}. We further investigate the adverse scenario, in which \emph{Eve}
%knows the transmission strategy and optimizes the number of transmit/recevie antennas
%correspondingly; for that case we find the worst-case achievable S.D.o.F..
%We also provide a method for constructing the precoding matrix pair achieving
%the maximum S.D.o.F.. Numerical results validate the efficacy of the proposed scheme.

This paper considers a scenario in which an \emph{Alice}-\emph{Bob} pair wishes to communicate in secret in the presence of an active \emph{Eve},
who is capable of jamming as well as eavesdropping in Full-Duplex (FD) mode.
As countermeasure, \emph{Bob} also operates in FD mode, using a subset of its antennas to act as receiver, and the remaining antennas to act as jammer and transmit noise.  With a goal to maximize the achievable secrecy degrees of freedom (S.D.o.F.) of the system, we provide the optimal
receive/transmit antennas allocation at \emph{Bob}, based on which
we determine in closed form the maximum achievable S.D.o.F.. We further investigate
the adverse scenario in which \emph{Eve} knows \emph{Bob}'s transmission
strategy and optimizes its transmit/receive antennas allocation in order to minimize the achievable S.D.o.F..
For that case we find the worst-case achievable
S.D.o.F.. We also provide a method for constructing the precoding
matrices of \emph{Alice} and \emph{Bob}, based on which the maximum S.D.o.F. can be achieved. Numerical results
validate the theoretical findings and demonstrate the performance of the proposed method in realistic settings.
\end{abstract}


%\section{Introduction}

\newtheorem{proposition}{Proposition}
\newtheorem{theorem}{Theorem}
\newtheorem{corollary}{Corollary}

\section{Introduction}
%Broadcast and superposition are two fundamental properties of the wireless medium.
%From the perspective of anti-eavesdropping,  % and improving secrecy rate performance
%broadcast nature makes wireless transmissions susceptible to eavesdropping,
%while superposition nature leads to the overlapping of the massage signals
%with interference signals at the legitimate receivers. Most of the existing literatures
%focus on giving countermeasures of eavesdropping, by utilizing multi-antenna techniques
%\cite{Jiangyuan10,Wornell11,Swindlehurst12}, or resorting to external helpers
%\cite{Zheng11,Han11,zheng151,LunDong10,Hoon14}, both with the idea of increasing
%the received signal-to-noise ratio (SNR) at the legitimate receiver and decreasing the received SNR at the eavesdropper.
%Recent works start to consider making use of the superposition nature and
%exploit interference to improve secrecy rate performance
%\cite{Ni14,Kalantari15,Lv15,Lingxiang162}, wherein instead of by transmitting extra jamming signals (artificial
%noise), which is power inefficient and may lower the overall network throughput,
%the authors try to degrade the received SNR at the eavesdropper by properly designing
%the co-channel interference and increasing the received interference power at the eavesdropper.

Communication security in the presence of malicious nodes has received a lot of attention.
Most of the current literature addresses the case in which the malicious nodes are
\emph{passive} eavesdroppers, i.e., they just listen. In that case, the eavesdroppers
reduce the secrecy rate by the rate they can sustain. Approaches to improve the secrecy rate in the presence of passive eavesdroppers include
multi-antenna techniques \cite{Jiangyuan10,JiangyuanLi12,Wornell11,Swindlehurst12} and artificial noise (jamming) based methods
\cite{Negi05,Swindlehurst11,Zheng11,Han11,zheng151,LunDong10,Hoon14,Gan13,LLX16}; all these methods
target at increasing the received signal-to-noise ratio (SNR) at the
legitimate receiver, or decreasing the received SNR at the eavesdropper.
Jamming can be implemented by
the source \cite{Negi05}, the external helper \cite{Swindlehurst11,Zheng11,Han11,zheng151,LunDong10,Hoon14}, or
the legitimate receiver who may work in Full-Duplex (FD) mode \cite{Gan13,LLX16}.
%The last two schemes have an advantage over the first one, since with them
%the message and jamming signals experience different channels, which provides channel diversity and thus
%improves the secrecy rate performance.
%Also, in comparison to cooperative jamming with external helpers, the FD scheme does not suffer from
%issues related to helpers' mobility, synchronization and trustworthiness.
%Due to the challenges in suppressing self-interference
%wireless communication systems like cellular have largely avoided FD.
%However, as short-range systems with low-power transmitters,
%such as small-cell systems and WiFi, are becoming dominant in the future wireless communication
%systems, the self-interference reduction problem will be much more manageable. This has
%sparked a renewed interest in FD \cite{Sabharwal14}.

Recently, the case of \emph{active} eavesdroppers has been receiving a lot of attention.
By active eavesdropper we here refer to a powerful adversary that can jam as well as eavesdrop the legitimate receiver.
One line of research in that area is gearing towards designing effective active attack schemes
for the purpose of minimizing the achievable secrecy transmission rate \cite{XZhou12,Mukherjee11,Andrey14}.
Another line of research focuses on %is from the perspective of anti-eavesdropping, i.e.,
detecting active attacks and offering countermeasures to guarantee reliable secret communications
\cite{Kapetanovic15,Azzam16,WYu15,Andrey15,Amitav13,QZhu11,Javan14}.
In particular, \cite{Kapetanovic15,Azzam16,WYu15} consider a massive multi-input multi-output (MIMO) scenario,
in which an active eavesdropper attacks the channel estimation process by transmitting artificial noise.
\cite{Andrey15,Amitav13},\cite{QZhu11,Javan14} consider a single-input single-output (SISO) scenario, a MIMO scenario, a relay
scenario, and an OFDM scenario, respectively, wherein an active eavesdropper
tries to reduce the total network throughput by choosing to be a jammer, or an eavesdropper,
or combination of the above, so that it creates the most unfavorable conditions for secret communications.
%a more unfavorable mode to the transmission process, i.e., jam or eavesdrop, or a combination of both.
To combat such malicious behavior, the source in \cite{Andrey15,Amitav13}
chooses between transmitting, remaining silent or acting as a jammer. %to transmit, remain silent, or act as a jammer.
The work of \cite{QZhu11,Javan14} conducts relaying selection and power allocation
among all the available sub-carriers, respectively.
%We should note that the works \cite{Holger13,Arsenia13} also study the wiretap channel in the presence of an
%active eavesdropper, but unlike

In this paper, we consider a MIMO \emph{Alice}-\emph{Bob}-\emph{Eve} wiretap channel,
in which \emph{Eve} is an active eavesdropper, who can transmit and receive in FD
fashion by appropriately allocating its antennas for transmission or reception.
Our goal is to provide countermeasures that will ensure maximum secrecy from the point of view of secrecy degrees of freedom (S.D.o.F.).
%as in \cite{Mukherjee11}, except that, unlike \cite{Mukherjee11}, which starts from the perspective of
%eavesdropping and aims to minimize the secrecy transmission rate via properly % the input covariance matrix
%transmission design at \emph{Eve}, our goal is to give an effective countermeasure to combat
%\emph{Eve} and improve reliable transmission rate.
%Our network comprises an \emph{Alice}-\emph{Bob} pair exchanging confidential messages, and an active \emph{Eve}
%who divides the antenna set into two parts, one part devoted to eavesdropping and the other to jamming.
Our main contributions are summarized as follows.

%assumes that \emph{Eve}
%has knowledge of the global instantaneous channel state information (CSI) including that of the main channel,
%we consider a more practical scenario in which only the number of antennas at each terminal is available.
%Moreover, the work of \cite{Mukherjee11}

\begin{enumerate}
\item As countermeasure, we proposed an FD \emph{Bob}, who transmits jamming signals while receiving.
%We should note that at the end of \cite{Mukherjee11}, a potential
%countermeasure is also proposed, where \emph{Alice} allocates some of its antennas to transmit jamming signals.
%The message and jamming signal experience different channels for the proposed
%FD \emph{Bob} based scheme, while experience common channels for that scheme proposed in \cite{Mukherjee11}.
%Since channel diversity is helpful to improve secrecy rate performance,
%the proposed FD \emph{Bob} based scheme has an advantage over that proposed by \cite{Mukherjee11}.
Under this scenario, we determine in closed form the maximum achievable S.D.o.F., as function
of the number of antennas at each terminal (see eq. (\ref{eq6})).
%i.e., the transmit/receive number of antennas at \emph{Eve}, the
%number of antennas at \emph{Alice}, and the number of antennas at \emph{Bob}.
Moreover, we give the optimal transmit/receive antenna allocation of \emph{Bob} (see (\ref{eqNbt})), which achieves the maximum S.D.o.F..
\item %Based on the closed form expression of the maximum achievable S.D.o.F., we
We obtain analytically the worst-case achievable S.D.o.F. (see eq. (\ref{eq8})), corresponding to the case in which
\emph{Eve} knows the strategy adopted by \emph{Alice} and
\emph{Bob} and optimizes its transmit/receive antenna allocation for the purpose of minimizing the
achievable S.D.o.F..
%We should note that the derivation of all these S.D.o.F. results, including the
%maximum achievable S.D.o.F. and also the worst-case achievable S.D.o.F.,
%requires no knowledge of any channel state information (CSI).
\item We provide a method for constructing the precoding matrix pair at \emph{Alice} and \emph{Bob}, %(see Section III. A),
which achieves the maximum S.D.o.F.. While the aforementioned achievable S.D.o.F. results
do not depend on channel state information (CSI), the precoding matrices depend on
the eavesdropping channels and also the null space of the self-interference
channels at \emph{Eve} and \emph{Bob}.
%We further show that when the self-interference channel at
%\emph{Bob} is unknown, the proposed construction still works but results in a smaller achievable S.D.o.F..
%In this way, we lower the requirements for the transmitter's ability in acquiring CSI.
\end{enumerate}

The rest of this paper is organized as follows. In
Section II, we describe the system model
and formulate the S.D.o.F. maximization problem. In Section III, we determine
in closed form the maximum achievable S.D.o.F., and provide an optimal transmission scheme
which achieves the maximum S.D.o.F..
In Section IV, we consider an active \emph{Eve}
who knows the transmission strategy adopted
by the legitimate terminals and tries to minimize the achievable S.D.o.F. by antenna allocation;
%adjusting the number of transmit/receive antennas;
for that case, we find the worst-case achievable S.D.o.F..
Numerical results are given in Section V and conclusions are drawn in Section VI.

\textit{Notation:}
$x\sim\mathcal{CN}(0,\Sigma)$ means $x$ is a random variable following a complex circular Gaussian
distribution with mean zero and covariance $\Sigma$; $(a)^+ \triangleq \max(a,0)$;
$\lfloor a\rfloor$ denotes the biggest integer which is less or equal to $a$;
$|a|$ denotes the absolute value of $a$.
We use lower case bold to denote vectors;
${\bf I}$ represents an identity matrix with appropriate size;
$\mathbb{C}^{N \times M}$ indicates a ${N \times M}$ complex matrix set;
${\bf{A}}^H$, $\rm{tr}\{\bf{A}\}$, $\rm{rank}\{\bf{A}\}$, and $|{\bf{A}}|$ stand for the hermitian transpose, trace,
rank and determinant of the matrix $\bf{A}$, respectively.

\begin{figure}[!t]
%\centering
\includegraphics[width=3in]{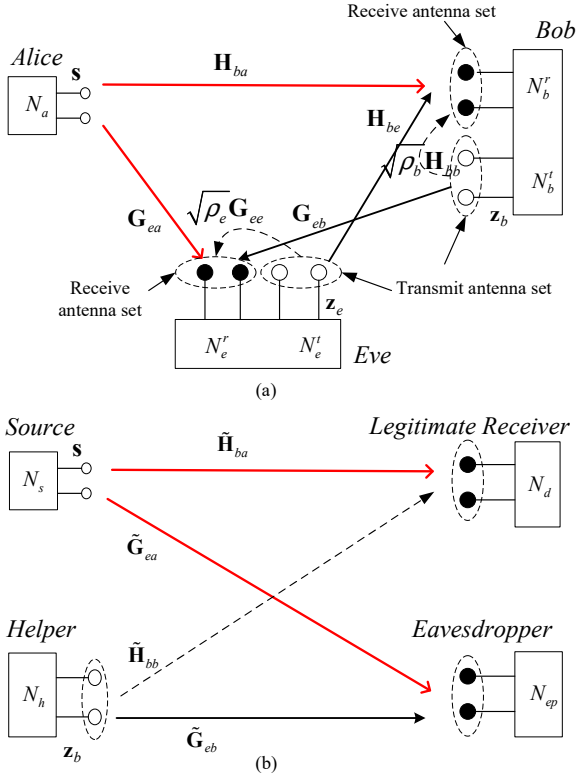}
 %where an .eps filename suffix will be assumed under latex,
% and a .pdf suffix will be assumed for pdflatex; or what has been declared via
\DeclareGraphicsExtensions.
\caption{(a) Gaussian wiretap channel with an active eavesdropper.
(b) Helper-assisted Gaussian wiretap channel with a passive eavesdropper.} % with respect to the maximum achievable S.D.o.F.
\vspace* {-12pt}
\end{figure}

\section{System Model and Problem Statement}
We consider a Gaussian wiretap channel (see Fig. 1(a)) consisting of \emph{Alice},
\emph{Bob}, and \emph{Eve}, equipped with $N_a$,
$N_b$ and $N_e$ antennas, respectively. \emph{Eve} is an active agent, who works in FD mode, i.e.,
it allocates $N_e^r$ antennas to receive signals and uses the remaining $N_e^t=N_e-N_e^r$
antennas to transmit isotropic noise, i.e., ${\bf z}_e$, with $E\{{\bf z}_e{\bf z}_e^H\} =({P}/{N_e^t}){\bf I}$.
\emph{Alice} wishes to send message
${\bf s} \sim \mathcal{CN}(\bf{0},\bf{I})$ to \emph{Bob} and keep it secret from
\emph{Eve}. Towards that objective, \emph{Bob} allocates $N_b^r$ antennas to receive the message
and uses the remaining $N_b^t=N_b-N_b^r$
antennas to transmit jamming signals, i.e., ${\bf z}_b$, with ${\bf z}_b \sim \mathcal{CN}(\bf{0},\bf{I})$. % , for the purpose of degrading \emph{Eve}'s channel.
Since \emph{Bob} transmits noise while receiving the signal of interest, he generates self-interference, and so does \emph{Eve}.
While several self-interference cancelation techniques have been reported, such as antenna isolation, analog-circuit-domain based methods and
digital-domain based methods, %however, today's state-of-the-art cannot achieve
full self-interference cancelation is still not achievable \cite{Sabharwal14}.
To describe the effect of residual self-interference
we employe the loop interference model of \cite{Gan13},
which quantifies the level of self-interference with a parameter $\rho \in [0,1]$, with $\rho=0$ denoting zero self-interference.
%with parameter $\rho_b ,\rho_e=0$ corresponding to the no self-interference case and $0<\rho_b ,\rho_e \le 1$ denoting different self-interference levels.

To improve the system performance, \emph{Alice} and \emph{Bob} will precode their transmissions, using precoding matrices ${\bf V}_a$ and ${\bf V}_b$, respectively.
The signal received at \emph{Bob} and \emph{Eve} can be respectively written as
\begin{align}
&{{\bf y}_b} = {\bf{H}}_{ba}{\bf V}_a{\bf s} + \sqrt{\rho_b}{{\bf{H}}_{bb}}{{\bf V}_b{\bf z}_b}+ {{\bf{H}}_{be}}{{\bf z}_e} + {{\bf{n}}_b }, \label{eq1} \\
&{{\bf{y}}_e} = {{\bf G}_{ea}}{\bf V}_a{\bf s} + {{\bf G}_{eb}}{{\bf V}_b{\bf z}_b}+ \sqrt{\rho_e}{{\bf G}_{ee}}{{\bf z}_e} + {{\bf{n}}_e}, \label{eq0}
\end{align}
where %${\bf V}_a$ and ${\bf V}_b$ are the precoding matrices at \emph{Alice} and \emph{Bob}, respectively;
${{\bf{n}}_b} \sim \mathcal{CN}(\bf{0},\bf{I}) $
and ${{\bf{n}}_e} \sim \mathcal{CN}(\bf{0},\bf{I})$
represent additive white Gaussian noise (AWGN) vectors at \emph{Bob} and \emph{Eve}, respectively;
${\bf{H}}_{ba}\in\mathbb{C}^{N_b^r \times N_a}$
and ${\bf{H}}_{be}\in\mathbb{C}^{N_b^r \times N_e^t}$ denote the channel matrices from
\emph{Alice} and \emph{Eve} to \emph{Bob}, respectively;
${\bf{G}}_{ea}\in\mathbb{C}^{N_e^r \times N_a}$ and ${\bf{G}}_{eb}\in\mathbb{C}^{N_e^r \times N_b^t}$
denote the channel matrices from \emph{Alice} and \emph{Bob} to \emph{Eve}, respectively;
${\bf{H}}_{bb}\in\mathbb{C}^{N_b^r \times N_b^t}$ and ${\bf{G}}_{ee}\in\mathbb{C}^{N_e^r \times N_e^t}$
represent the self-interference channel matrices at \emph{Bob} and \emph{Eve}, respectively;
$\rho_b$ and $\rho_e$ denote the self-interference level of \emph{Bob} and \emph{Eve}, respectively.
The transmitted signals including the message signal $\bf s$ and the jamming signals ${\bf z}_b$ and ${\bf z}_e$
are independent of each other, and independent of the noise ${\bf n}_b$ and ${\bf n}_e$.
Since \emph{Alice} and \emph{Bob} are not expected to cooperate with \emph{Eve}, \emph{Eve} cannot do any precoding.
The only way \emph{Eve} can affect the achievable S.D.o.F. is by optimizing its transmit/receive antenna allocation.
%Since \emph{Alice}/\emph{Bob} will not cooperate with \emph{Eve}, and as it will become clear that only the
%number of transmit/receive antennas will affect the achievable S.D.o.F.,
%we assume that \emph{Eve} does not do any precoding but it may optimize its number of transmit/receive antennas.

In the above, the Gaussian signaling assumption is made in order to maximize the achievable secrecy transmission rate \cite{Liu09, Liu10}.
Also, the flat fading assumption used in (\ref{eq1}), (\ref{eq0})
is valid %for the scenario in which all frequency components of the signal will experience the same magnitude of fading, which arises
when the coherence bandwidth of the channel is larger than the bandwidth of the transmitted signal \cite{David06}.
Here we assume that all channels are known at the legitimate nodes, including the CSI for \emph{Eve}. This is possible
in situations in which \emph{Eve} is an active network user and its whereabouts and behavior can be monitored.

%In this paper, we make the following assumptions:
%\begin{enumerate}
%\item We assume Gaussian signaling for the transmit antennas of \emph{Bob} and \emph{Eve}. Thus, the effective noise at
%both receivers is Gaussian. In this case, a Gaussian input signal at \emph{Alice} is the optimal choice for the purpose
%of achieving the maximum secrecy transmission rate \cite{Liu09, Liu10}.
%\item All the channels are flat fading and the corresponding matrices are full rank. Global CSI
%is available at the legitimate nodes, including the CSI for \emph{Eve}. This is possible in situations in which
%\emph{Eve} is an active network user and her whereabouts and behavior can be monitored.
%\end{enumerate}

%All channels are assumed to be flat fading. In order to reveal the fundamental limit of anti-eavesdropping, we assume that
%global channel state information (CSI) is available, including the CSI for \emph{Eve}.

For a given precoding matrix pair $({\bf V}_a, {\bf V}_b)$,
the maximum achievable rate at \emph{Bob} and \emph{Eve} can be respectively expressed as \cite{OggierBabak11}
\begin{subequations}
\begin{align}
& R_b = {\rm {log}}|{\bf I}+({\bf I}+{\bf W}_b)^{-1}{\bf{H}}_{ba}{\bf Q}_a{\bf{H}}_{ba}^H|, \label{eq3a}\\ %
& R_e = {\rm {log}}|{\bf I}+({\bf I}+{\bf W}_e)^{-1}{\bf{G}}_{ea}{\bf Q}_a{\bf{G}}_{ea}^H|, \label{eq3b}
\end{align}
\end{subequations}
where ${\bf{Q}}_a\triangleq{\bf V}_a{\bf V}_a^H$ and ${\bf Q}_b\triangleq{\bf V}_b{\bf V}_b^H$
denote the input covariance matrices at \emph{Alice} and \emph{Bob}, respectively, with
the average transmit power budget ${\rm tr}\{{\bf{Q}}_a\}={\rm tr}\{{\bf{Q}}_b\}=P$;
the interference covariance matrices at \emph{Bob} and \emph{Eve} respectively are
\begin{align}
&{\bf W}_b\triangleq \rho_b {\bf{H}}_{bb}{\bf Q}_b {\bf{H}}_{bb}^H+\frac{P}{N_e^t}{\bf{H}}_{be}{\bf{H}}_{be}^H, \nonumber\\
&{\bf W}_e\triangleq{\bf{G}}_{eb}{\bf Q}_b {\bf{G}}_{eb}^H+\frac{\rho_e P}{N_e^t}{\bf{G}}_{ee}{\bf{G}}_{ee}^H. \nonumber
\end{align}

Correspondingly,
%the achievable secrecy rate is
%\begin{align}
%R_s=(R_b-R_e)^+. \label{eq0}
%\end{align}
the achievable S.D.o.F., representing the high SNR behavior of the achievable secrecy rate \cite{Liang09}, is % as the SNR approaches infinity
\begin{align}
d_{s,a}({\bf{Q}}_a,{\bf{Q}}_b) \triangleq \mathop{\lim }\limits_{ P \to \infty }
\dfrac{R_b-R_e}{{\rm log} \ P}, \label{eq4}
\end{align}
provided that a positive secrecy rate can be achieved.

The goal of this paper is to determine the maximum achievable S.D.o.F. over the transmission schemes at \emph{Alice} and
\emph{Bob}, i.e., the antenna allocation at \emph{Bob} and the precoding matrices of \emph{Alice} and
\emph{Bob}. To that goal, in the following,
we will first determine the optimal number of transmit/receive antennas at \emph{Bob}, based on which we then analytically determine the
maximum achievable S.D.o.F.. %, as a function of the number of antennas at each terminal.
%the number of transmit/receive antennas at \emph{Eve},
%the number of antennas at \emph{Bob}, and the number of antennas at \emph{Alice}.
Subsequently, we find the worst-case achievable S.D.o.F. for the adverse scenario, in which
\emph{Eve} is smart and tries to minimize the achievable S.D.o.F. by adjusting
the number of transmit/receive antennas.

%We should note that for the case of $N_e \ge N_b$, the best choice for \emph{Eve} is to
%allocate a number of $N_b$ antennas to transmit, in which case, no positive
%S.D.o.F. can be achieved.
%In this paper, we only need to study the nontrivial case of $N_e < N_b$.

\section{The Maximum Achievable S.D.o.F.}
In \cite{Lingxiang16,Lingxiang162}, we determined the maximum achievable S.D.o.F. for a helper-assisted Gaussian wiretap channel,
which consists of a source equipped with $N_s$ antennas, a
legitimate receiver equipped with $N_d$ antennas, a passive eavesdropper equipped with $N_{ep}$ antennas,
and an external helper (sending jamming signals to confuse \emph{Eve}) equipped with $N_h$ antennas.
In that scenario, the main idea for achieving the maximum S.D.o.F. is to include into the source and helper precoding matrix pair
the maximum possible linearly precoding vector pairs along which the message and jamming signals are
aligned into the same received subspace of \emph{Eve}, subject to the constraint that
the total number of signal streams \emph{Bob} can see is no greater than its total number of receive antennas.
The achievable S.D.o.F. equals the number of precoding vectors that has been included into the source precoding matrix.
%For ease of exposition, let ${\rm WT}(N_s, N_h, N_d, N_{ep})$ denote
For easy reference the helper-assisted Gaussian wiretap channel studied in \cite{Lingxiang16} is depicted in Fig. 1(b).
%As compared with \cite{Lingxiang16}, each receiver in Fig. 1(a),
%including the legitimate receiver and also the malicious receiver, suffers from the noise sent by \emph{Eve}.
As we will show next, the maximum achievable S.D.o.F. of the wiretap channel
of Fig. 1(a) is equal to that of the wiretap channel of  Fig. 1(b) with parameters as given in the following
proposition.

\begin{proposition}
\textit{Provided that $N_e^t<\min\{N_b^r, N_e^r\}$, the maximum achievable S.D.o.F. of the
MIMO Gaussian wiretap channel of Fig. 1(a), is equal to that of a helper-assisted wiretap channel of Fig. 1(b),
with $N_s=N_a$, $ N_h=N_b^t$, $N_d=N_b^r-N_e^t$ and
$N_{ep}=N_e^r-N_e^t$.}
\end{proposition}
\begin{proof}
See Appendix A.
\end{proof}

\emph{Remark 1:} Based on \emph{Proposition 1}, one can see that if $N_e^t<\min\{N_b^r, N_e^r\}$
the maximum S.D.o.F. of the system under consideration can be determined based on results on the helper-assisted wiretap channel.
Otherwise, if $N_e^t \ge N_b^r$ and independent of $N_e^r$,
the maximum achievable S.D.o.F. is zero, since \emph{Bob} already cannot see any interference-free subspaces;
if $N_e^t \ge N_e^r$, \emph{Eve} cannot see any interference-free subspaces, and so the maximum achievable S.D.o.F. is equal to $\min\{(N_b^r-N_e^t)^+, N_a\}$. Therefore, for the purpose of computing the maximum achievable S.D.o.F. of the system under consideration, we only need
to investigate that of the corresponding helper-assisted wiretap channel.

Next, we show that for a fixed total number of helper and destination antennas, i.e., $ N_h+N_d=N_{\rm sum}$,
one can find a solution for the number of helper antennas which achieves the maximum S.D.o.F..
%$ N_h+N_d=N_{\rm sum}$, by optimizing $N_h$ we can further maximize the achievable S.D.o.F..
Details are given in the following proposition.

\begin{proposition}
\textit{Consider the helper-assisted wiretap channel of Fig. 1(b).
Suppose that $N_h$ and $N_d$ can vary but their sum is always fixed at $N_{\rm sum}$.
%Let $N_h$ start increasing and $N_d$ decreasing in a way that $N_h+N_d$ in constant and equal to $N_{\rm sum}$.
%$N_h$ vary from 1 to $N_{\rm sum}$, subject to $N_{\rm sum}=N_h+N_d$.
Then, the maximum achievable S.D.o.F. is
\begin{align}
{\small d_{s,p} = \min\{\delta, N_{\rm sum}, N_s \},  \label{eq5}}
\end{align}
where $\delta \triangleq \lfloor \frac{(N_{\rm sum}-|N_s-N_{ep}|)^+}{3}\rfloor+(N_s-N_{ep})^+$.
%Moreover, let $(\hat N_h, \hat N_d)$ denote an optimal solution of the number of transmit/receive antennas at {Bob}.
\begin{enumerate}
\item If $N_{\rm sum} \le N_{ep}-N_s$, the maximum
achievable S.D.o.F. is zero for any pair of $(N_h, N_d)$. %transmit/receive antennas at {Bob}.
\item If $N_{\rm sum} \le N_s-N_{ep}$, the maximum S.D.o.F. is achieved when
$N_d=N_{\rm sum}$ with no antennas being allocated to the helper.
%we have one solution as follows, $\hat N_d=N_{\rm sum}$ and $\hat N_h=0$.
\item If $N_{\rm sum} > |N_s-N_{ep}|$, the maximum S.D.o.F. is achieved
when $N_h=\hat N_h$, where
\begin{align}
\hat N_h=\left\{ {\begin{array}{*{20}{c}}
{N_{ep}-N_s+\lfloor \frac{N_{\rm sum}-|N_s-N_{ep}|}{3}\rfloor}&{\rm if}\ { N_s \le N_{ep}},\\
\lfloor \frac{N_{\rm sum}-|N_s-N_{ep}|}{3}\rfloor &{\rm if}\ {N_s > N_{ep}},
\end{array}} \right. \nonumber
\end{align}
and the remaining $N_{\rm sum}-\hat N_h$ antennas are assigned to the legitimate receiver.%, with $i \triangleq N_b-3 \lfloor \frac{N_b-|N_s-N_e|}{3}\rfloor$.
\end{enumerate}}
\end{proposition}
\begin{proof}
See Appendix B.
\end{proof}

Combining \emph{Proposition 1} and \emph{Proposition 2}, we can determine
the maximum achievable S.D.o.F. for the system under consideration as follows.

\begin{figure}[!t]
\centering
\includegraphics[width=3in]{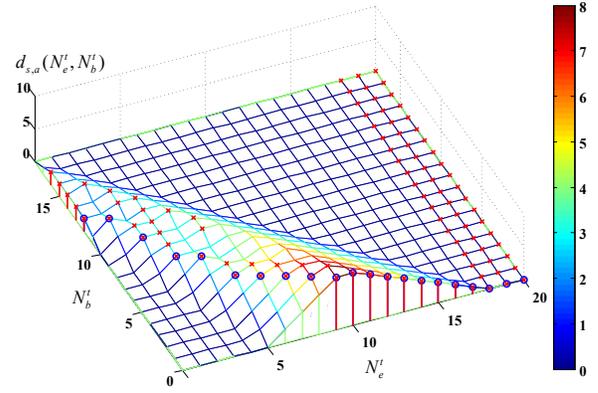}
 %where an .eps filename suffix will be assumed under latex,
% and a .pdf suffix will be assumed for pdflatex; or what has been declared via
\DeclareGraphicsExtensions. \caption{The maximum achievable S.D.o.F. for the system with $N_a=10$, $N_b=18$ and $N_e=20$.}
\vspace* {-6pt}
\end{figure}

\begin{theorem}
\textit{Consider a MIMO Gaussian wiretap channel, as depicted in Fig. 1(a).
%Using the optimal number of transmit/receive antennas at {Bob},
The maximum achievable S.D.o.F. is
\begin{align}
\small d_{s,a}(N_e^t) = \left\{ {\begin{array}{*{20}{c}}
\min\{(N_b-N_e^t)^+, N_a\} &{\rm if}\ {N_e^t \ge N_e^r}, \\
\min\{\eta, (N_b-N_e^t)^+, N_a\}&{\rm if}\ {N_e^t < N_e^r},
\end{array}} \right. \label{eq6}
\end{align}
with $\eta \triangleq \lfloor \frac{(N_b-N_e^t-|N_a-N_e^r+N_e^t|)^+}{3}\rfloor+(N_a-N_e^r+N_e^t)^+$.
The maximum S.D.o.F. is achieved when Bob uses ${N_b^t}^\star$ antennas to transmit,
with ${N_b^t}^\star$ given in (\ref{eqNbt}) at the top of the next page, and the remaining
$N_b-{N_b^t}^\star$ antennas receive.}
\end{theorem}

%when $N_e^t \ge N_e^r$ or $N_e^t \ge N_b$,
%the achievable S.D.o.F. will not benefit from jamming signals by \emph{Bob}, and so
%we set ${N_b^t}^\star=0$. Let us consider the
%case of $N_e^t < N_e^r$ and $N_e^t < N_b$, in which the maximum S.D.o.F. is achieved when
%$ N_b^r >N_e^t $. On applying \emph{Proposition 1},
%one can see that for any given $ N_b^r >N_e^t $, the maximum achievable
%S.D.o.F. is equal to that of a helper-assisted wiretap channel with the number of
%antennas $N_s=N_a$, $ N_h=N_b^t$, $N_d=N_b^r-N_e^t$ and $N_{ep}=N_e^r-N_e^t$.
%Substituting these values into \emph{Proposition 2}, we
%arrive at the optimal number of transmit antennas at \emph{Bob},
%${N_b^t}^\star$, given in (\ref{eqNbt}) at the top of the page.

\begin{figure*}[!t]
\normalsize
\setcounter{equation}{6}
\begin{align}
{N_b^t}^\star = \left\{ {\begin{array}{*{20}{c}}
N_e^r-N_e^t-N_a+\lfloor \dfrac{N_b-N_e^t-|N_a-N_e^r+N_e^t|}{3}\rfloor&
{\rm if}\ N_e^t < \min\{N_e^r, N_b-|N_a-N_e^r+N_e^t|\} \ {\rm and}\ N_a \le N_e^r-N_e^t\\
\lfloor \dfrac{N_b-N_e^t-|N_a-N_e^r+N_e^t|}{3}\rfloor &
{\rm if}\ N_e^t < \min\{N_e^r, N_b-|N_a-N_e^r+N_e^t|\} \ {\rm and}\ N_a > N_e^r-N_e^t \\
{0}&{\rm otherwise}
\end{array}} \right. \label{eqNbt}
\end{align}
\setcounter{equation}{7}
\hrulefill
%\vspace*{-5pt}
\end{figure*}

\begin{proof}
See Appendix C.
\end{proof}

\emph{Theorem 1} provides the number of transmit antennas at Bob which achieves the maximum S.D.o.F..
This is is illustrated in Fig. 2, where we plot the maximum achievable S.D.o.F. for the system with $N_a=10$, $N_b=18$ and $N_e=20$.
Specifically, for a given antenna number pair $(N_e^t, N_b^t)$, we plot the
achievable S.D.o.F. based on \emph{Remark 1}. For each fixed $N_e^t$, we find,
with the numerical search method, the points which achieve the maximum S.D.o.F., and mark them with red crosses.
Looking at the slice of the graph corresponding to a fixed $N_e^t$,
one can see that there are one or more $N_b^t$'s which achieve the maximum S.D.o.F.,
and ${N_b^t}^\star$ marked by a blue circle, coincides with one of those red crosses.
%for each $N_e^t$, with the red cross we mark the maximum achievable S.D.o.F. by different $N_b^t$'s, and with blue circle
%we mark the S.D.o.F. achieved by choosing the transmit antenna number of \emph{Bob} as in (\ref{eqNbt}).
%Fig. 2 confirms that for each given $N_e^t$, there are one or multiple $N_b^t$'s which achieve the
%maximum S.D.o.F., and ${N_b^t}^\star$ is one of them.

\subsection{The proposed transmission scheme which achieves the maximum S.D.o.F.}
%When \emph{Bob} has a total of $N_b$ antennas, there is an optimal number of transmit antennas, referred to as
%${N_b^t}^\star$, that will maximized the achievable S.D.o.F..
%In this subsection, we will first determine the optimal number of transmit antennas at
%\emph{Bob}, denoted by ${N_b^t}^\star$. Correspondingly, the number of receive antennas at \emph{Bob} is $N_b-{N_b^t}^\star$.
%In this section we provide the expression of ${N_b^t}^\star$, and also a method for constructing the precoding matrix pair,
%denoted by $({\bf V}_a^\star, {\bf V}_b^\star)$,
%which achieves the maximum S.D.o.F..

With the optimal allocation of trnasmit/receive antennas at \emph{Bob}, we next
construct the pair $({\bf V}_a^\star, {\bf V}_b^\star)$ which achieves the maximum S.D.o.F..
\begin{enumerate}
\item For the case of ${N_b^t}^\star=0$, and along the lines of Appendix A,
one can see that the wiretap channel of Fig. 1(a) is equivalent to
a classic three-node wiretap channel, with the main channel and eavesdropping channel
being equal to ${{\bf U}_{b}^0}^H{\bf{H}}_{ba}$ and $ {{\bf U}_{e}^0}^H{\bf{G}}_{ea}$, respectively. Here,
${\bf U}_{b}^{0}$ and $ {\bf U}_{e}^{0}$ are the orthonormal basis of the null space of ${\bf H}_{be}$
and ${\bf G}_{ee}$, respectively. Therefore, by applying the precoding matrix design
of the three-node wiretap channel of \cite{Wornell11}, the maximum
S.D.o.F. can be achieved. According to \cite{Wornell11}, the precoding matrices are
constructed by selecting those linearly independent precoding vectors along which
the legitimate channel has better quality than the eavesdropping channel.
\item For the case of ${N_b^t}^\star \ne 0$, and along the lines of Appendix A, one can see that
the wiretap channel of Fig. 1(a) is equivalent to
a classic helper-assisted wiretap channel, with the channels to \emph{Bob} being equal to ${ {\bf U}_{b}^0}^H{\bf{H}}_{ba}$
and $ {{\bf U}_{b}^0}^H{\bf{H}}_{bb}$, the channels to \emph{Eve}
being equal to $ {{\bf U}_{e}^0}^H{\bf{G}}_{ea}$ and $ {{\bf U}_{e}^0}^H{\bf{G}}_{eb}$,
and the number of antennas being $N_s=N_a$, $ N_h=N_b^t$, $N_d=N_b^r-N_e^t$ and $N_{ep}=N_e^r-N_e^t$.
Therefore, by applying the precoding matrix design of \cite{Lingxiang16,Lingxiang162}
to this equivalent helper-assisted wiretap channel, the maximum S.D.o.F. can be achieved.
The main idea here is to select the maximum possible number of linearly independent precoding vector pairs along which the message and
jamming signals are aligned into the same received subspace of \emph{Eve}.
In particular, we divide the candidate set of precoding vector pairs into three subsets, i.e., C1, in which
the message signal sent by \emph{Alice} spreads within the null space of the eavesdropping channel,
C2, in which the message does not spread within the null space of the eavesdropping channel and
\emph{Bob} is self-interference free, and C3, in which the message does not spread
within the null space of the eavesdropping channel and
\emph{Bob} suffers from self-interference. We select precoding vector pairs from C1 first, followed by
C2 and then C3, until there are no more
candidate precoding vector pairs or the total number of signal streams \emph{Bob}
can see is equal to its total number of receive antennas.
For more details on determining the number of candidates of each subset and their formulas,
please refer to \cite{Lingxiang16,Lingxiang162}.
It is worth noting that (to be used in Section V) the formulas of the precoding vector pairs in C1 only depend on the channel matrix $ {{\bf U}_{e}^0}^H{\bf{G}}_{ea}$; the formulas of the precoding vector pairs in C3 only depend on the channel matrices
$ {{\bf U}_{e}^0}^H{\bf{G}}_{ea}$ and $ {{\bf U}_{e}^0}^H{\bf{G}}_{eb}$; in addition to $ {{\bf U}_{e}^0}^H{\bf{G}}_{ea}$ and $ {{\bf U}_{e}^0}^H{\bf{G}}_{eb}$, the formulas of the precoding vector pairs in C2
also depend on the channel matrix $ {{\bf U}_{b}^0}^H{\bf{H}}_{bb}$.
\end{enumerate}

\section{Worst-Case Achievable S.D.o.F. in the Presence of A Smart \emph{Eve}}
%Since legitimate transmitter and receiver
%have no knowledge about how such an active wiretapper will
%influence the channel conditions, they have to be prepared for
%the worst, i.e.,.
In this section, we consider a scenario in which \emph{Eve} knows the transmit strategies
at both \emph{Alice} and \emph{Bob}, and therefore it derives $d_{s,a}(N_e^t) $, based on
which it adjusts the number of its transmit antennas in order to minimize the
achievable S.D.o.F., i.e., $d_{s,a}(N_e^t) $. In that case, the worst-case maximum achievable S.D.o.F. is
\begin{align}
d_{s,a}^{\rm{wc}} = \mathop{\min }\limits_{ 0 \le N_e^t \le N_e } d_{s,a}(N_e^t). \label{eq7}
\end{align}

\begin{theorem}
\textit{Consider the MIMO Gaussian wiretap channel of Fig. 1(a).
Assume that {Eve} knows the transmit strategies at {Alice} and {Bob}.
Then, the maximum achievable S.D.o.F. is
given in (\ref{eq8}), which is shown at the top of next page.}
\end{theorem}
\begin{proof}
See Appendix D.
\end{proof}

\emph{Theorem 2} enables us to make some interesting observations, which are given in the following Corollaries.
\begin{corollary}
\textit{For the purpose of minimizing the achievable S.D.o.F., {Eve} will jam or eavesdrop, but will not adopt a combination
of both.}
\end{corollary}
\begin{proof}
From the proof of {\emph {Theorem 2}} in Appendix D, one can see that the minimum
value of $d_{s,a}(N_e^t)$ is obtained only when $N_e^t =0 $ or $N_e^t = N_e $.
This completes the proof.
\end{proof}

\begin{corollary}
\textit{If $N_b > N_e$, a positive S.D.o.F. can always be achieved with the proposed cooperative transmission scheme.}
\end{corollary}
\begin{proof}
With the expression of (\ref{eq8}), it can be verified that the worst-case achievable
S.D.o.F. is greater than zero for the case of $N_b > N_e$. This completes
the proof.
\end{proof}

\begin{figure*}[!t]
\normalsize
\setcounter{equation}{8}
\begin{align}
d_{s,a}^{\rm{wc}} = \left\{ {\begin{array}{*{20}{c}}
{0}&{\rm if}\ {N_e \ge N_b },\\
\min\{ \lfloor  \dfrac{N_b-N_e+N_a}{3}  \rfloor, N_b-N_e, N_a\}&{\rm if}\ \max\{\dfrac{N_b-N_a}{2}, N_a\} \le N_e < N_b, \\
\min\{ \lfloor \dfrac{N_b-N_a+N_e}{3}  \rfloor+N_a-N_e, N_b-N_e\} &{\rm if}\ {\dfrac{N_b-N_a}{2} \le N_e < \min\{ N_b, N_a\}}\ {\rm and}\ {N_e > N_a- N_b}, \\
N_b-N_e &{\rm if}\ {\dfrac{N_b-N_a}{2} \le N_e < \min\{ N_b, N_a\}}\ {\rm and}\ {N_e \le N_a- N_b}, \\
{N_a}&{\rm if}\ {N_e < \min\{\dfrac{N_b-N_a}{2}, N_b \} }.
\end{array}} \right. \label{eq8}
\end{align}
\setcounter{equation}{9}
\hrulefill
%\vspace*{-5pt}
\end{figure*}

\section{Numerical Results}
%We examine the achievable secrecy transmission rate of the proposed scheme and validate
%its robustness in the position of \emph{Eve} along
%the $y$-coordinate, the self-interference level, and also the channel state information.
As already mentioned, the achievable S.D.o.F. reveals the high SNR behavior of the achievable secrecy rate.
%in practical communication systems the SNR may be medium or low.
% where the achievable secrecy rate is the key performance metric.
In this section, we consider a more realistic SNR scenario, and
demonstrate the secrecy rate performance of the proposed approach. In particular, we consider a scenario as shown in Fig. 3.
%We consider a system model as illustrated in Fig. 2.
\emph{Alice} and \emph{Bob} are respectively fixed at coordinates $(-R,0)$ and $(R,0)$ (unit: meters). %, where $R$ is the distance parameter.
The smaller the $R$, the higher the received SNR at \emph{Bob} will be.
\emph{Eve} can move in one of the following two ways, i.e., parallel to the $x$-axis and between the points $(-20,-R)$ and $(20,-R)$, and
parallel to the $y$-axis and between the points $(0,10)$ and $(0,0)$.

\begin{figure}[!t]
\centering
\includegraphics[width=3in]{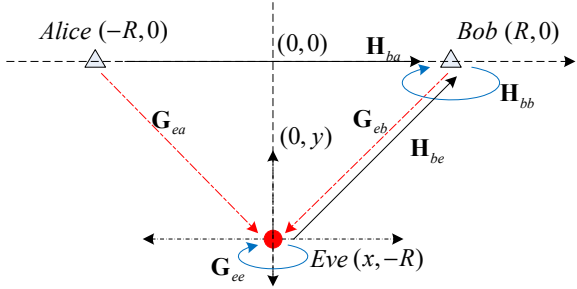}
 %where an .eps filename suffix will be assumed under latex,
% and a .pdf suffix will be assumed for pdflatex; or what has been declared via
\DeclareGraphicsExtensions. \caption{Model used for numerical experiments.}
\vspace* {-6pt}
\end{figure}

Unless otherwise specified, we consider the strong self-interference level $\rho_b=\rho_e=\rho=1$,
and we set $N_a=4$, $N_b=7$, $N_e^t=1$ and $N_e^r=5$. The transmit power of each node is $P= 0$dBm.
The noise power level is set as $\sigma^2=-60$dBm.
The power is equally allocated between different signal streams at each node.
According to \emph{Theorem 1}, for the above system, the maximum achievable S.D.o.F. of 2 can be achieved by choosing $N_b^t=2$, $N_b^r=5$.
Setting $N_b^t=2$, $N_b^r=5$, and according to Section III. A, one can see that
the system under consideration is equivalent to the helper-assisted wiretap channel of Fig. 1(b),
with the number of antennas being $N_s=4$, $ N_h=2$, $N_d=4$ and $N_{ep}=4$;
for that helper-assisted wiretap channel, the number of candidate precoding vector pairs in
C1, C2 and C3 are respectively 0, 0 and 2.
Following the construction method of Section III. A and since $N_d=4$ and for each precoding vector pair in C3 \emph{Bob} suffers from self-interference,
%we select the two candidate precoding vector pairs in C3.
%Since $N_d=2$, we cannot pick any more precoding vector pairs without violating the constraint that
%the total number of signal streams \emph{Bob} can see is no greater than its total number of receive antennas. Concluding,
we can select two precoding vector pairs in C3 without violating the constraint that
the total number of signal streams \emph{Bob} can see is no greater than its total number of receive antennas. Therefore,
a total of two precoding vector pairs can be picked, and as such a number
of two message signal streams will be sent from \emph{Alice}.
We construct the precoding matrix pair assuming
exact knowledge of the channels.

With the precoding matrix pair,
we examine the achievable secrecy transmission rate, i.e., $(R_b-R_e)^+$,
where $R_b$ and $R_e$ are given by (\ref{eq3a}) and (\ref{eq3b}), respectively \cite{OggierBabak11}.
Results are obtained based on $1,000$ Monte Carlo runs.
In each run, the effect of the channel on the transmitted signal
is modeled by a multiplicative scalar of the form $d^{-c/2}e^{j\theta}$ \cite{Inaltekin09}, where
$d$ is the distance between the transmit and receive terminals, $c$ is the path loss exponent and $\theta$ is a random phase,
which is taken to be uniformly distributed
within $[0, 2 \pi)$ and independent between runs.
The value of $c$ is typically in the range of 2 to 4. In our simulations we set $c=3.5$.
We assume that the distance of different combinations of transmit-receive antennas corresponding to the same
link is the same, and as such the corresponding path loss is the same.
%Unless otherwise specified, we consider the strong self-interference level $\rho_b=\rho_e=\rho=1$.

For comparison, we also plot the average achievable secrecy rate of
the half-duplex (HD) scheme, wherein \emph{Bob} receives with all of its antennas.
For the HD scheme, the precoding matrix of \emph{Alice} consists of
the generalized eigenvectors corresponding to the largest two generalized eigenvalues of the matrix pair \cite{Wornell11}
\begin{align}
\small (\hat{\bf{H}}_{ba}^H({\bf I}+\frac{P}{N_e^t}\hat{\bf{H}}_{be}\hat{\bf{H}}_{be}^H)^{-1}\hat{\bf{H}}_{ba},
\hat{\bf{G}}_{ea}^H({\bf I}+\frac{\rho_e P}{N_e^t}\hat{\bf{G}}_{ee}\hat{\bf{G}}_{ee}^H)^{-1}\hat{\bf{G}}_{ea}),  \label{eqMP}
\end{align}
where $\hat{\bf{H}}_{ba}$ and $\hat{\bf{H}}_{be}$ denote the channel matrices to \emph{Bob},
$\hat{\bf{G}}_{ea}$ and $\hat{\bf{G}}_{ee}$ represent the channel matrices to \emph{Eve}.
From Section III. A, the proposed transmission scheme in terms of the achievable
S.D.o.F. can be either equivalent with a three-node wiretap channel
when ${N_b^t}^\star=0$, or equivalent with a helper-assisted wiretap channel when ${N_b^t}^\star \ne 0$.
In the former case, the proposed scheme reduces to an HD scheme. In the latter case,
the proposed scheme always achieves a greater S.D.o.F..
For comparison fairness, in the HD scheme we consider selecting the same number of message signal streams as in the proposed scheme.

Figs. 4 and 5 illustrate the average achievable secrecy transmission rate
as function of \emph{Eve}'s position, with the $x$-coordinate
varying from $-20$ to $20$ and the $y$-coordinate fixed at $-R$. %, for the case of $R=10$ and $R=1$, respectively.
Fig. 4 corresponds to $R=10$, which represents a low SNR scenario for \emph{Bob}, while Fig. 5 corresponds to $R=1$,
which is a high SNR scenario for \emph{Bob}.
From Fig. 4, one can see that the proposed FD scheme performs overall better than the HD scheme,
except when \emph{Eve} is to the left of \emph{Alice} or to the right of \emph{Bob}.
The behavior in the latter cases should be expected, since when \emph{Eve} is to the left of \emph{Alice},
the received jamming signal is too weak to disturb \emph{Eve}'s channel.
%help impede the eavesdropping channel due to the large-scale fading,
As a result, the HD scheme, which uses all of \emph{Bob}'s antennas to receive, performs better.
When \emph{Eve} is to the right of \emph{Bob}, the received SNR is already small
even if \emph{Bob} does not send jamming signals, and as a result,
the HD scheme also performs better.
Naturally, for the higher SNR case, the advantage of the proposed FD approach is bigger and evident over the entire range (see Fig. 5).
To illustrate the secrecy rate advantage of using the proposed antenna allocation at \emph{Bob}, i.e., $N_b^t=2$ and $N_b^r=5$,
in Fig. 5 we also plot the achievable secrecy transmission rate for another allocation, i.e., $N_b^t=3$ and $N_b^r=4$;
in that case and according to Section III. A, one can see that only an S.D.o.F. of 1 can be achieved.
As expected, the achievable secrecy transmission rate
of that latter case is almost half of the proposed case, for which an S.D.o.F. of 2 can be achieved.

\begin{figure}[!t]
\centering
\includegraphics[width=3in]{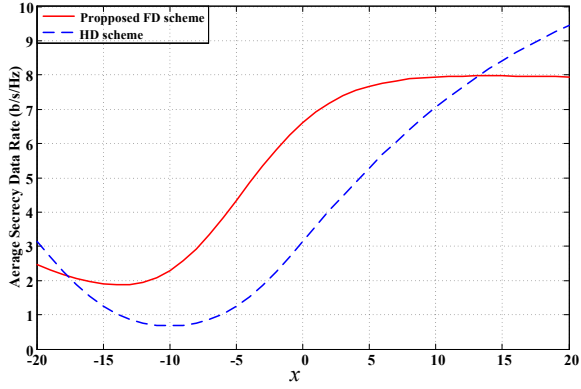}
 %where an .eps filename suffix will be assumed under latex,
% and a .pdf suffix will be assumed for pdflatex; or what has been declared via
\DeclareGraphicsExtensions. \caption{Average achievable secrecy rate versus the position of \emph{Eve} along the $x$-coordinate. The distance
parameter $R=10$.}
\vspace* {-6pt}
\end{figure}

\begin{figure}[!t]
\centering
\includegraphics[width=3in]{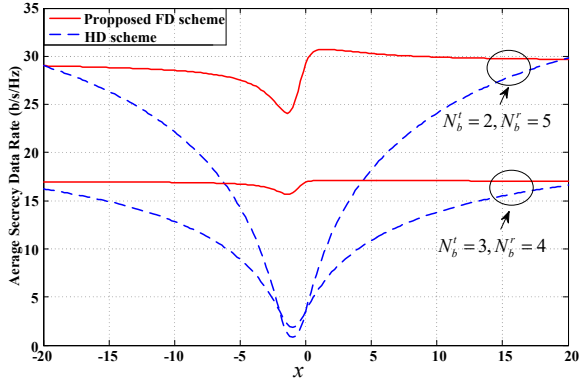}
 %where an .eps filename suffix will be assumed under latex,
% and a .pdf suffix will be assumed for pdflatex; or what has been declared via
\DeclareGraphicsExtensions. \caption{Average achievable secrecy rate versus the position of \emph{Eve} along the $x$-coordinate. The distance
parameter $R=1$.}
\vspace* {-6pt}
\end{figure}

In Fig. 6, we plot the average achievable secrecy transmission rate
versus the position of \emph{Eve} along the $y$-axis, for the case of $R=10$ and $R=5$.
The figure shows that for both cases, the achievable secrecy transmission
rate of the proposed FD scheme remains constant for all positions of \emph{Eve}. In contrast, the achievable secrecy transmission rate
of the HD scheme decreases as $y$ approaches zero.
This can be explained as follows. As \emph{Eve} comes closer
to \emph{Alice}, it receives a stronger signal, and as a result the secrecy rate of the HD scheme decreases.
On the other hand, in the proposed FD scheme, the message signal sent by \emph{Alice}
and the jamming signal sent by \emph{Bob} are aligned into the same received subspace of \emph{Eve}, thus keeping
\emph{Eve}'s eavesdropping capability constant, and as a result, keeping the achievable secrecy rate of
the proposed FD scheme constant.

\begin{figure}[!t]
\centering
\includegraphics[width=3in]{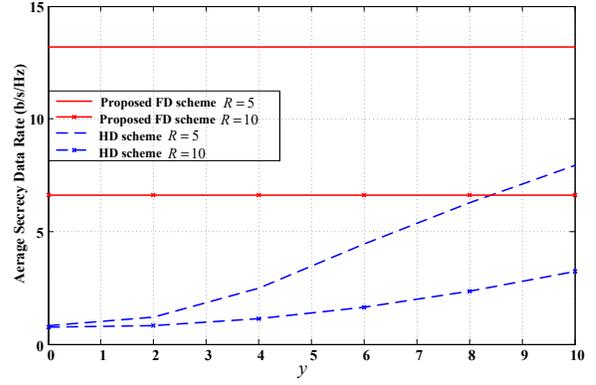}
 %where an .eps filename suffix will be assumed under latex,
% and a .pdf suffix will be assumed for pdflatex; or what has been declared via
\DeclareGraphicsExtensions. \caption{Average achievable secrecy rate versus the position of \emph{Eve} along the $y$-coordinate.}
\vspace* {-6pt}
\end{figure}

\begin{figure}[!t]
\centering
\includegraphics[width=3in]{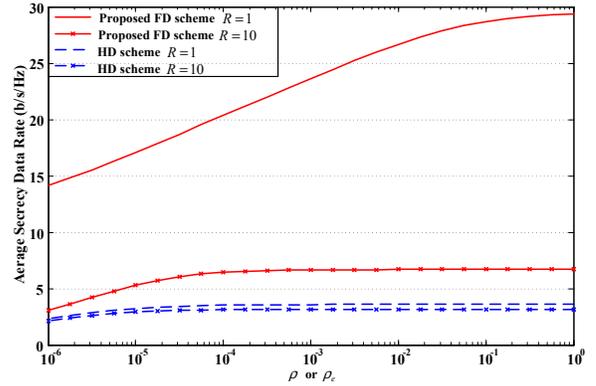}
 %where an .eps filename suffix will be assumed under latex,
% and a .pdf suffix will be assumed for pdflatex; or what has been declared via
\DeclareGraphicsExtensions. \caption{Average achievable secrecy rate versus the self-interference level.}
\vspace* {-3pt}
\end{figure}

%\begin{figure}[!t]
%\centering
%\includegraphics[width=3in]{rhoVaries_Na=4Nbt=1Nbr=5Net=1Ner=5}
% %where an .eps filename suffix will be assumed under latex,
%% and a .pdf suffix will be assumed for pdflatex; or what has been declared via
%\DeclareGraphicsExtensions. \caption{Average achievable secrecy rate versus the self-interference level.
%$N_a=4$, $N_b^t=N_e^t=1$, $N_b^r=N_e^r=5$.}
%\vspace* {-3pt}
%\end{figure}

Fig. 7 illustrates the average achievable secrecy transmission rate of the proposed scheme
as function of the self-interference level $\rho$, and that of the HD scheme as
function of the self-interference level $\rho_e$, for the case of $R=10$ and $R=1$.
We should note that since for the HD scheme \emph{Alice} determines its precoding matrix with (\ref{eqMP}),
the achievable secrecy transmission rate only relates to $\rho_e$.
One can see that the achievable secrecy rate of the FD scheme increases as $\rho$ increases.
This is because, by aligning the message and jamming signals into
%This can be explained as follows. The proposed FD scheme aligns the message signal and the jamming signal into
the same received subspace of \emph{Eve}, the proposed scheme delivers a distorted message signal to \emph{Eve}, which
makes the eavesdropping channel more sensitive to self-interference.
Therefore, the achievable secrecy rate of the FD scheme increases with increasing level of self-interference.
While the achievable secrecy rate of the HD scheme %in which \emph{Bob} uses all of its antennas to receive,
also increases with increasing level of the self-interference at \emph{Eve},
the increase is small as compared to the proposed scheme.

%Fig. 7 plots the achievable secrecy rate results for the scenario with
%$N_a=4$, $N_b^t=4$, $N_b^r=3$, $N_e^t=1$ and $N_e^r=5$, where we precode the message and jamming signals with
%a precoding matrix pair consisting of two precoding vector pairs from C2.
%By the definition of C2, \emph{Bob} is self-interference free, and as a result,
%the achievable secrecy rate is independent of the self-interference at \emph{Bob}.
%Next, we will show the average achievable secrecy transmission rate for another scenario in which
%the self-interference at \emph{Bob} matters.
%In Fig. 8, we set $N_a=4$, $N_b^t=N_e^t=1$, $N_b^r=N_e^r=5$, which is equivalent to the helper-assisted
%wiretap channel, with the number of antennas being $N_s=4$, $ N_h=1$, $N_d=4$ and $N_{ep}=4$;
%for that system the number of candidate precoding vector pairs of C1, C2 and C3 are respectively 0, 0 and 1.
%For the proposed scheme we select the only one candidate precoding vector pair in C3, along which \emph{Bob}
%suffers from self-interference. For the HD scheme, we select
%the generalized eigenvector corresponding to the largest generalized eigenvalue of the matrix pair in (\ref{eqMP}).
%The observations in Fig. 8 are similar to those of Fig. 7.

In order to separately check the effect of the self-interference level, i.e., $\rho_b$ or $\rho_e$, on the achievable secrecy rate
performance of the proposed scheme, in Fig. 8, we set $\rho_e=10^{-3}$ and plot the average achievable secrecy transmission rate
versus the self-interference level $\rho_b$; also, we set $\rho_b=10^{-3}$ and plot the average achievable secrecy transmission rate
versus the self-interference level $\rho_e$.
One can see that the achievable secrecy transmission rate decreases slightly with $\rho_b$,
while it increases drastically with $\rho_e$.
This can also be explained by the fact that,
for the FD scheme the eavesdropping channel is more sensitive to self-interference.

\begin{figure}[!t]
\centering
\includegraphics[width=3in]{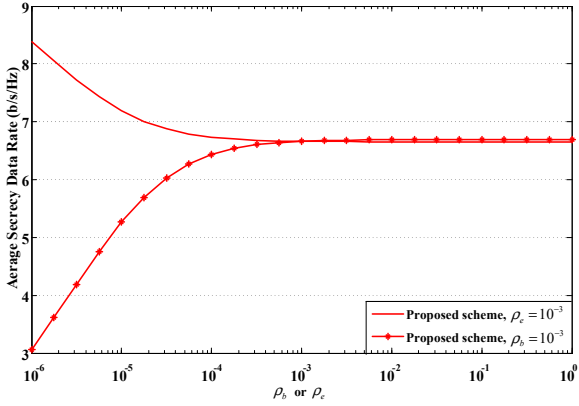} %RhoE=001rhoVaries_xR=10Na=4Nbt=1Nbr=5Net=1Ner=5
 %where an .eps filename suffix will be assumed under latex,
% and a .pdf suffix will be assumed for pdflatex; or what has been declared via
\DeclareGraphicsExtensions. \caption{Average achievable secrecy rate versus the self-interference level.
The distance parameter $R=10$.}
\vspace* {-6pt}
\end{figure}

In practice, perfect channel estimates are difficult to obtain.
Since the proposed precoding matrix design highly depends on the
channels, we next examine the secrecy rate performance in the presence of imperfect channel estimates.
%In Fig. 9, we plot the achievable secrecy rate with imperfect CSI of the eavesdropper's channels.
{We model imperfect CSI through a Gauss-Markov uncertainty of the form \cite{Nosrat11}
%The channel from \emph{S}$_i$ ($i=1,2$) to \emph{E} is
\begin{align}
{\bf G}_{ei}=d_{ei}^{-c/2}\left(\sqrt{1-\alpha^2}\bar{\bf G}_{ei}+\alpha\Delta\bar{\bf G}_{ei}\right), i=a, b, \label{eqImpCSI}
\end{align}
where $0 \le \alpha \le 1$ denotes the channel uncertainty. $\alpha=0$ and $\alpha=1$
correspond to perfect channel knowledge and no CSI knowledge, respectively.}
%$\bar{\bf G}_{i}$ represents the estimated channel part at \emph{S}$_i$.
The entries of $\bar{\bf G}_{ei}$ are $e^{j\theta}$ with $\theta$ be a random phase
uniformly distributed within $[0, 2 \pi)$.
$\Delta\bar{\bf G}_{ei}\sim \mathcal{CN}(\bf{0},\bf{I})$
represents the Gaussian error channel matrices.
$d_{ei}$ denotes the distance from \emph{Alice} or \emph{Bob}.
With the same channel model as in (\ref{eqImpCSI}), we model the channel uncertainty of the channels ${\bf H}_{bi}$, $i=a,b,e$.
We construct the precoding matrix pair $({\bf V}_a, {\bf V}_b)$ with the estimated channels.

In Fig. 9, we plot the achievable secrecy rate with respect to the channel uncertainty in ${\bf H}_{bi}$, $i=a,b,e$,
for the proposed antenna allocation scheme, i.e., $N_b^t=2$, $N_b^r=5$.
It can be observed that the achievable secrecy rate
remains constant for different channel uncertainties of ${\bf H}_{bi}$, $i=a,b,e$.
This should be expected, since the constructed precoding matrix pair consists of two
precoding vector pairs from C3, whose formulas only depend on
the matrices $ {{\bf U}_{e}^0}^H{\bf{G}}_{ea}$ and $ {{\bf U}_{e}^0}^H{\bf{G}}_{eb}$.
Therefore, the channels ${\bf H}_{bi}$, $i=a,b,e$ do not
enter in the construction of the precoding matrix pair.
Indeed, for the equivalent helper-assisted wiretap channel with the antenna allocation given by \emph{Proposition 2}, i.e.,
$\hat N_h$, it can be verified that
there are no candidate precoding vector pairs in C2. Therefore, the achievable secrecy rate of proposed scheme
is independent of the channel uncertainties of ${\bf H}_{bi}$, $i=a,b,e$.
As illustrated in Fig. 2, for a given fixed $N_e^t$ there may be more than one $N_b^t$'s which can achieve
the maximum S.D.o.F..
Intuitively, those schemes achieving the same S.D.o.F. can also achieve the same secrecy rate performance,
which, combined with the fact that the proposed schemes's
achievable secrecy rate remains unchanged even when the channel estimates turns noisy,
indicates that the proposed scheme will outperform the others.
Next, with simulations we show that advantage of the proposed scheme.
Let's take the antenna allocation, i.e., $N_b^t=4$, $N_b^r=3$, as an example.
Substituting $N_b^t=4$, $N_b^r=3$ into Section III. A, one can see that the maximum S.D.o.F. of 2 can also be achieved.
In particular, with $N_b^t=4$, $N_b^r=3$ the system under consideration is equivalent to the helper-assisted wiretap channel of Fig. 1(b),
with the number of antennas being $N_s=4$, $ N_h=4$, $N_d=2$ and $N_{ep}=4$;
for that helper-assisted wiretap channel, the number of candidate precoding vector pairs in
C1, C2 and C3 are respectively 0, 2 and 2.
Following the construction method in Section III. A, we first select the two candidate precoding vector pairs in C2.
Since $N_d=2$, we cannot pick any more precoding vector pairs without violating the constraint that
the total number of signal streams \emph{Bob} can see is no greater than its total number of receive antennas. Concluding,
a total of two precoding vector pairs can be picked from C2, and as such an S.D.o.F. of 2 can be achieved \cite{Lingxiang16,Lingxiang162}.
Based on Fig. 9 one can see that the proposed scheme, i.e., $N_b^t=2$, $N_b^r=5$, and
that with $N_b^t=4$, $N_b^r=3$, provide the same secrecy rate performance when the channel estimates are perfect.
Moreover, when the channel estimates are noisy, i.e., $\alpha >0$, the proposed scheme outperforms the other one, since
the achievable secrecy rate of the proposed scheme remains unchanged while
that of the other scheme drops with the increase of uncertainty in the channels ${\bf H}_{bi}$, $i=a,b,e$.
This is because, unlike the proposed scheme the formulas of the precoding vector pairs of the other one are from C2, and as such they
depend on the channel $ {{\bf U}_{b}^0}^H{\bf{H}}_{bb}$.

On the other hand, in Fig. 9 it can be observed that the achievable secrecy rate drops with the
increase of uncertainty in the channels ${\bf G}_{bi}$, $i=a,b,e$.
This should be expected, since the benefits brought by the proposed scheme
come from the successful alignment of the message and jamming signals at \emph{Eve}.
To achieve that goal, the exact knowledge of the channels ${\bf G}_{ei}$, $i=a,b,e$, is necessary.
As a conclusion, one can see that the uncertainty in the channels ${\bf G}_{ei}$, $i=a,b,e$,
is more dangerous.

\begin{figure}[!t]
\centering
\includegraphics[width=3in]{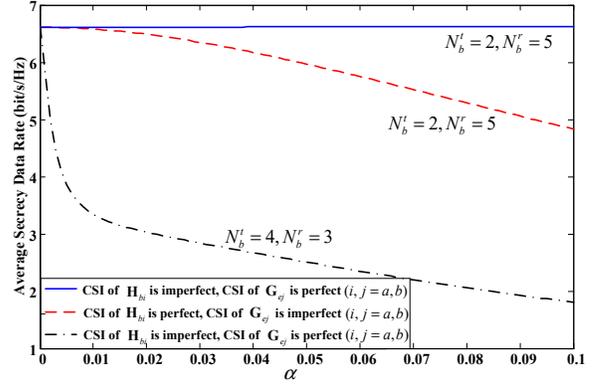} %CSIerrorR=10rho=1Na=4Nbt=4Net=1Nbr=3Ner=5
 %where an .eps filename suffix will be assumed under latex,
% and a .pdf suffix will be assumed for pdflatex; or what has been declared via
\DeclareGraphicsExtensions. \caption{Average achievable secrecy rate versus channel uncertainty. The distance
parameter $R=10$.}
\vspace* {-6pt}
\end{figure}

%\begin{figure}[!t]
%\centering
%\includegraphics[width=3in]{CSIerrorR=10rho=1Na=4Nbt=1Nbr=5Net=1Ner=5}
% %where an .eps filename suffix will be assumed under latex,
%% and a .pdf suffix will be assumed for pdflatex; or what has been declared via
%\DeclareGraphicsExtensions. \caption{Average achievable secrecy rate versus channel uncertainty.
%$R=10$, $N_a=4$, $N_b^t=N_e^t=1$, $N_b^r=N_e^r=5$.}
%\vspace* {-6pt}
%\end{figure}

\section{Conclusion}
We have analytically addressed the S.D.o.F. maximization problem
of a MIMO Gaussian wiretap channel in the presence of an active \emph{Eve}.
%who divides the antenna set into two parts, one part devoted to eavesdropping and the other to jamming.
Specifically, we have proposed a Full-Duplex \emph{Bob} scheme, where \emph{Bob} divides the antenna set into two
parts, one devoted to receiving and the other to jamming.
%cooperative secrecy transmission scheme, where \emph{Bob} works in the
%FD mode and aligns the signals with that from \emph{Alice} into a same received subspace of \emph{Eve}.
%It has been proved that the proposed scheme can achieve the maximum S.D.o.F..
Based on the proposed scheme, we have derived the optimal number of transmit/receive antennas at \emph{Bob}, and determined
the maximum S.D.o.F., as a function of the number of antennas at each terminal.
%, i.e., the number of transmit/receive antennas at \emph{Eve}, the number of antennas at \emph{Bob},
%and the number of antennas at \emph{Alice}.
We have further found the worst-case achievable S.D.o.F. for the adverse scenario in which \emph{Eve}
knows the transmit strategies and tries to minimize the S.D.o.F. by adjusting its number of transmit/receive antennas.
Our analysis has revealed that a positive S.D.o.F. can be guaranteed as long as it holds that $N_b>N_e$.
%Moreover, from the perspective of minimizing the achievable S.D.o.F., a smart \emph{Eve} will jam or eavesdrop,
%instead of doing a combination of both.
We have also constructed a precoding matrix pair which achieves the maximum S.D.o.F..
Numerical results have revealed the advantages of the proposed
secrecy transmission scheme over the existing half-duplex scheme, %. Also, the results have validated
and have validated the robustness of the proposed scheme under realistic scenarios.
%in the position of \emph{Eve} along the $y$-axis, the
%self-interference level, and also the state information of the channels to \emph{Bob}.
%
%In this work, we have assumed that the total transmit power of each node is equally allocated between multiple signal streams,
%and the precoding matrix pair derived is asymptotically optimal in the high SNR regimes.
%Optimal power allocation between multiple signal streams for the purpose of
%improving of the secrecy rate performance in the low SNR regimes needs more future work.

\appendices

\section{Proof of \emph{Proposition 1}} \label{appA}
Given an arbitrary point $({\bf V}_a,{\bf V}_b)$, with
${\rm{tr}} \{{\bf Q}_a\}=P$ and $ {\rm{tr}} \{{\bf Q}_b\}=P$.
We can respectively rewrite ${\bf Q}_a$ and ${\bf Q}_b$ as
${\bf Q}_a=  P \bar{\bf Q}_a$ and ${\bf Q}_b=  P \bar{\bf Q}_b $,
with ${\rm{tr}} \{\bar{\bf Q}_a\}={\rm{tr}}\{{\bar{\bf Q}_b}\} = 1$.
Correspondingly, (\ref{eq3a}) can be rewritten as
\begin{align}
R_b=I_b^2-I_b^1,\label{eqb1}
\end{align}
where
\begin{subequations}
\begin{align}
& I_b^1 \triangleq {\rm {log}}|{\bf I}+P{\bf M}{\bf{H}}_{bb}\bar{\bf Q}_b{\bf{H}}_{bb}^H|, \label{eqb2a} \\
& I_b^2 \triangleq {\rm {log}}|{\bf I}+P{\bf M}({\bf{H}}_{bb}\bar{\bf Q}_b{\bf{H}}_{bb}^H
+{\bf{H}}_{ba}\bar{\bf Q}_a{\bf{H}}_{ba}^H)|, \label{eqb2b}
\end{align}
\end{subequations}
with ${\bf M} \triangleq ({\bf I}+\dfrac{P}{N_e^t}{\bf{H}}_{be}{\bf{H}}_{be}^H)^{-1}$.

Let ${\bf{H}}_{be}{\bf{H}}_{be}^H=\left[{\bf U}_{b}^1\ {\bf U}_{b}^0 \right]
\left[ {\begin{array}{*{20}{c}}
{{{\bf{\Sigma}}_{b}}}&{\bf{0}}\\
{\bf{0}}&{{{\bf{0}}}}
\end{array}} \right]
\left[ {\begin{array}{*{20}{c}}
{{{\bf{U}}_{b}^{1H}}}\\
{{{\bf{U}}_{b}^{0H}}}
\end{array}} \right]$ be the singular value decomposition (SVD), and then
\begin{align}
{\bf M} ={\bf U}_{b}^1({\bf I}+\dfrac{P}{N_e^t}{\bf{\Sigma}}_{b})^{-1}{\bf U}_{b}^{1H}+{\bf U}_{b}^0{\bf U}_{b}^{0H}. \label{eqb3}
\end{align}
Substituting (\ref{eqb3}) into (\ref{eqb2a}) and (\ref{eqb2b}), respectively, we obtain
{\small
\begin{subequations}
\begin{align}
&\mathop{\lim }\limits_{ P \to \infty }\dfrac{I_b^1}{{\rm {log}}(P)} = \mathop{\lim }\limits_{ P \to \infty }\dfrac{{\rm {log}}|{\bf I}+P\bar{\bf{H}}_{bb}\bar{\bf Q}_b\bar{\bf{H}}_{bb}^H|}{{\rm {log}}(P)}, \label{eqb4}\\
&\mathop{\lim }\limits_{ P \to \infty }\dfrac{I_b^2}{{\rm {log}}(P)} = \mathop{\lim }\limits_{ P \to \infty }\dfrac{{\rm {log}}|{\bf I}+P(\bar{\bf{H}}_{bb}\bar{\bf Q}_b\bar{\bf{H}}_{bb}^H
+\bar{\bf{H}}_{ba}\bar{\bf Q}_a\bar{\bf{H}}_{ba}^H)|}{{\rm {log}}(P)}, \label{eqb5}
\end{align}
\end{subequations}}
where $\bar{\bf{H}}_{bb} \triangleq {\bf U}_{b}^{0H}{\bf{H}}_{bb}$, $\bar{\bf{H}}_{ba} \triangleq  {\bf U}_{b}^{0H}{\bf{H}}_{ba}$.

Combining (\ref{eqb1}), (\ref{eqb4}) and (\ref{eqb5}), we arrive at that
{\small\begin{align}
\mathop{\lim }\limits_{ P \to \infty }\dfrac{R_b}{{\rm {log}}(P)} = \mathop{\lim }\limits_{ P \to \infty }\dfrac{{\rm {log}}|{\bf I}+({\bf I}+P\bar{\bf{H}}_{bb}\bar{\bf Q}_b\bar{\bf{H}}_{bb}^H)^{-1}P\bar{\bf{H}}_{ba}\bar{\bf Q}_a\bar{\bf{H}}_{ba}^H|}{{\rm {log}}(P)}. \label{eqb6}
\end{align}}

Letting ${\bf{G}}_{ee}{\bf{G}}_{ee}^H=\left[{\bf U}_{e}^1\ {\bf U}_{e}^0 \right]
\left[ {\begin{array}{*{20}{c}}
{{{\bf{\Sigma}}_{e}}}&{\bf{0}}\\
{\bf{0}}&{{{\bf{0}}}}
\end{array}} \right]
\left[ {\begin{array}{*{20}{c}}
{{{\bf{U}}_{e}^{1H}}}\\
{{{\bf{U}}_{e}^{0H}}}
\end{array}} \right]$ be the SVD, and applying the same derivations from (\ref{eqb1}) to (\ref{eqb6}), we obtain that
{\small\begin{align}
\mathop{\lim }\limits_{ P \to \infty }\dfrac{R_e}{{\rm {log}}(P)} = \mathop{\lim }\limits_{ P \to \infty }\dfrac{{\rm {log}}|{\bf I}+({\bf I}+P\bar{\bf{G}}_{eb}\bar{\bf Q}_b\bar{\bf{G}}_{eb}^H)^{-1}P\bar{\bf{G}}_{ea}\bar{\bf Q}_a\bar{\bf{G}}_{ea}^H|}{{\rm {log}}(P)}, \label{eqb7}
\end{align}}
where $\bar{\bf{G}}_{ea}\triangleq {\bf U}_{e}^{0H}{\bf{G}}_{ea}$ and $\bar{\bf{G}}_{eb}\triangleq  {\bf U}_{e}^{0H}{\bf{G}}_{eb}$.

Combining (\ref{eqb6}) and (\ref{eqb7}), one can see that
the achievable S.D.o.F. is equal to that of a helper-assisted wiretap channel,
with the channels to \emph{Bob} as $ {\bf U}_{b}^{0H}{\bf{H}}_{ba}$
and $ {\bf U}_{b}^{0H}{\bf{H}}_{bb}$, and the channels to \emph{Eve}
as $ {\bf U}_{e}^{0H}{\bf{G}}_{ea}$ and $ {\bf U}_{e}^{0H}{\bf{G}}_{eb}$, respectively.
Since $N_e^t<N_b^r$ and $N_e^t<N_e^r$, and all the channel matrices are assumed to be
full rank, this helper-assisted wiretap channel has effective number of antennas
$N_s=N_a$, $ N_h=N_b^t$, $N_d=N_b^r-N_e^t$ and
$N_{ep}=N_e^r-N_e^t$.
This completes the proof.

\section{Proof of \emph{Proposition 2}} \label{appB}
It can be verified that, for the case of $N_{\rm sum} \le N_s-N_{ep}$, the maximum achievable S.D.o.F. equals $N_{\rm sum}$,
which is consistent with (\ref{eq5}); for the case of $N_{\rm sum} \le N_{ep}-N_s$, the maximum achievable S.D.o.F. equals 0,
which is also consistent with (\ref{eq5}).
Thus, in the sequel, we only need to focus on the case of $N_{\rm sum} > |N_s-N_{ep}|$, in which
\begin{align}
d_{s,p} = \min\{\delta, N_{\rm sum}, N_s \}, \label{eqa0}
\end{align}
where $\delta=\lfloor \frac{N_{\rm sum}-|N_s-N_{ep}|}{3}\rfloor+(N_s-N_{ep})^+$.

%Since here \emph{Bob} operates in FD mode and \emph{Eve} works in passive mode (allocates all
%of its antennas to receive), the wiretap channel in Fig. 1 reduces to a helper-assisted one.
According to \emph{Theorem 1} of \cite{Lingxiang16} or equation (36) of \cite{Lingxiang162}, the maximum achievable S.D.o.F. for such a
helper-assisted wiretap channel is
\begin{align}
g(N_h)=\min \{d_{c=1}(N_h)+d_{c=2}^\star(N_h),N_d,N_s\}, \label{eqa1}
\end{align}
where
\begin{subequations}
\begin{align}
&d_{c=1}(N_h)\triangleq(N_s-N_{ep})^+ +s_1(N_h), \label{eqa2a}\\
&d_{c=2}^\star(N_h)\triangleq \min \{s_2(N_h), \lfloor (N_d-d_{c=1}(N_h))^+ / 2 \rfloor\}, \label{eqa2b}
\end{align}
\end{subequations}
with
\begin{align}
&s_1(N_h) \triangleq (\min\{N_s, N_{ep}\}+\min\{(N_h-N_d)^+,N_{ep}\}-N_{ep})^+, \nonumber \\
& s_2(N_h) \triangleq (\min\{N_s, N_{ep}\}+\min\{N_h,N_{ep}\}-N_{ep})^+-s_1(N_h). \nonumber
\end{align}

In the following, we will consider two distinct cases, i.e., the case of $N_s \le N_{ep}$
and the case of $N_s > N_{ep}$. For each case we first give a specific value of $N_h$, denoted by $\hat N_h$,
which satisfies $g(\hat N_h)=d_{s,p}$. We then prove that for any $N_h \ne \hat N_h$,
it holds that $g( N_h)\le d_{s,p}$.
In this way, we complete the proof of \emph{Proposition 2}.

\subsection{For the case of $N_s \le N_{ep}$}
It holds that $\delta= \lfloor \dfrac{N_{\rm sum}-|N_s-N_{ep}|}{3}\rfloor$.

Let $\hat N_d=2\lfloor \dfrac{N_{\rm sum}-|N_s-N_{ep}|}{3}\rfloor+i$, and
\begin{align}
\hat N_h=\lfloor \dfrac{N_{\rm sum}-|N_s-N_{ep}|}{3}\rfloor+(N_{ep}-N_s), \label{eqa3}
\end{align}
where $i \triangleq N_{\rm sum}-3 \bar {N}_d$. By definition $i \in \{0,1,2\}$.
\\\\
{\emph{A. 1 When $\delta \ge N_s$} }

In this subcase, it can be verified that $N_{\rm sum} \ge N_s$. Thus, (\ref{eqa0}) becomes % and $\delta \ge N_s$
\begin{align}
d_{s,p}=N_s. \label{eqa54}
\end{align}
On the other hand, since $\hat N_h \ge N_{ep}$, (\ref{eqa2a}) becomes
\begin{align}
d_{c=1}(\hat N_h)=N_s. \label{eqa6}
\end{align}
Substituting (\ref{eqa6}) into (\ref{eqa1}) and combined with the fact
that $\min\{ \hat N_d, N_s\}=N_s$, we arrive at $g(\hat N_h)=N_s$. Besides,
by (\ref{eqa1}) the inequality $g( N_h) \le N_s$ always holds true.
Therefore, the maximum value of $g( N_h)$ over $N_h$ is
\begin{align}
g(\hat N_h)=N_s \mathop = \limits^{(a)} d_{s,p}, \nonumber
\end{align}
where (a) comes from the equality in (\ref{eqa54}).
\\\\
{\emph{A. 2 When $\delta < N_s$} }

In this subcase, it can be verified that $\delta <N_{\rm sum} $. Thus, (\ref{eqa0}) becomes % and $\delta < N_s$
\begin{align}
d_{s,p}=\delta . \label{eqa55}
\end{align}

On the other hand, since $N_s \le N_{ep}$ and $\hat N_h- \hat N_d \le N_{ep}-N_s$, (\ref{eqa2a})
and (\ref{eqa2b}) respectively becomes
\begin{align}
{d_{c=1}(\hat N_h)=0, \label{eqa8}} \\
{d_{c=2}^\star(\hat N_h)=\delta . \label{eqa8add}}
\end{align}
%$d_{c=1}(\hat N_h)=(N_s+\min \{\hat N_h- \hat N_d,N_{ep}\}-N_{ep})^+$, which indicates that
%\begin{align}
%\small {d_{c=1}(\hat N_h)=\left\{ {\begin{array}{*{20}{c}}
%{{N_s}}&{\rm if}\ {\hat N_h - \hat N_d \ge {N_{ep}}},\\
%\lfloor \frac{N_{\rm sum}-|N_s-N_{ep}|}{3}\rfloor+i &{\rm if}\ {\hat N_h - \hat N_d < {N_{ep}}}.
%\end{array}} \right. \label{eqa8}}
%\end{align}
Substituting (\ref{eqa8}) and (\ref{eqa8add}) into (\ref{eqa1}) and combined with the fact that
$\min\{\delta, \hat N_d, N_s\}=\delta$, we obtain
\begin{align}
g(\hat N_h)=\delta\mathop = \limits^{(a)} d_{s,p}, \label{eqaghat}
\end{align}
where (a) comes from the equality in (\ref{eqa55}).

\emph{Next, we will prove that for any other $N_h \ne \hat N_h$ it holds that
$g( N_h)\le d_{s,p}$, thus completing the proof that
the maximum value of $g( N_h)$ over $N_h$ is $g(\hat N_h)=d_{s,p}$.
To achieve that goal, we introduce $\bar N_d=\lfloor \dfrac{N_{\rm sum}-|N_s-N_{ep}|}{3}\rfloor$, and
\begin{align}
\bar N_h=2\lfloor \dfrac{N_{\rm sum}-|N_s-N_{ep}|}{3}\rfloor+i+(N_{ep}-N_s). \nonumber
\end{align}
With similar derivations from (\ref{eqa54}) to (\ref{eqaghat}) it can be verified that
$g(\bar N_h)=d_{s,p}=g(\hat N_h)$. In the remaining text of this subsection, we will
show that for any other $N_h \ne \bar N_h$ it holds that
$g( N_h)\le d_{s,p}$.}

i) For any $N_h > \bar N_h$, it holds that $N_d < \bar N_d$.
In addition, by (\ref{eqa1}) it holds that $g(N_h) \le N_d$. Therefore,
\begin{align}
g(N_h)< \bar N_d= d_{s,p}. \nonumber
\end{align}

ii) For any $N_h < \bar N_h$, say $N_h = \bar N_h-k$ with $k \ge 1$, i.e.,
\begin{align}
&N_h =2 \bar N_d+i+(N_{ep}-N_s) -k, \nonumber\\ %\label{eqa9a}\\
&N_d = \bar N_d +k. \nonumber %\label{eqa9b}
\end{align}
Thus, $N_h -N_d=\bar N_d+(N_{ep}-N_s)+i-2k < N_{ep}$,
which, together with (\ref{eqa2a}), gives
\begin{align}
d_{c=1}(N_h)=(\bar N_d +i-2k)^+.  \label{eqa56}
\end{align}

\begin{enumerate}
\item For the case of $2k \le \bar N_d+i$, (\ref{eqa56}) becomes
$d_{c=1}(N_h)=\bar N_d+i-2k$, which, combined with (\ref{eqa2b}),
gives $d_{c=2}^\star (N_h)\le \lfloor \dfrac{3k-i}{2}\rfloor$. Therefore,
\begin{align}
&g( N_h) \le d_{c=1}(N_h)+d_{c=2}^\star(N_h) \nonumber \\
&\le \bar N_d +i-2k +\lfloor \dfrac{3k-i}{2}\rfloor  \nonumber \\
& \mathop  \le \limits^{(a)} \bar N_d \mathop  = \limits^{(b)} d_{s,p}.  \nonumber  %\label{eqa10}
\end{align}
Here, since $i \le 2$ and $k \ge 1$, it holds true that $i-2k +\lfloor \dfrac{3k-i}{2}\rfloor \le 0$,
and as a result, (a) holds true; (b) comes from the equality in (\ref{eqa55}).
\item For the case of $\bar N_d+i < 2k \le 2(\bar N_d+1)$, (\ref{eqa56}) becomes
$d_{c=1}(N_h)=0$. In addition, by (\ref{eqa2b}), it holds that $d_{c=2}^\star (N_h)\le \lfloor N_d/2 \rfloor$,
which, combined with $N_d=\bar N_d +k \le 2\bar N_d +1$, indicates that $d_{c=2}^\star (N_h)\le \bar N_d$. Therefore,
\begin{align}
&g( N_h) \le d_{c=2}^\star(N_h) \le \bar N_d = d_{s,p}.  \nonumber  %\label{eqa10}
\end{align}
\item For the case of $k \ge \bar N_d+2$, (\ref{eqa56}) becomes
$d_{c=1}(N_h)=0$. Therefore,
\begin{align}
&g( N_h) \le d_{c=2}^\star(N_h) \le s_2(N_h) \nonumber \\
&=\min\{N_s, N_s+N_h-N_{ep}\}  \nonumber \\
&\le 2 \bar N_d+i-k \le \bar N_d+i-2 \nonumber \\
& \le \bar N_d = d_{s,p}.  \nonumber%\label{eqa10}
\end{align}
\end{enumerate}

Based on the above two subcases, i.e., \emph{A. 1} and \emph{A. 2}, one can see that for the case of $N_s \le N_{ep}$
the maximum value of $g(N_h)$ over $N_h$ is $g(\hat N_h)=g(\bar N_h)=d_{s,p}$.
It is worth noting that, although both $\hat N_h$ and $\bar N_h$ can achieve the maximum S.D.o.F.,
as it can be observed in Section V, for the helper-assisted
wiretap channel with the antenna allocation given by $\hat N_h$, the formulas of the candidate precoding vector pairs
are independent of the channel matrices to \emph{Bob}. Therefore, when the channel estimates are noisy the proposed scheme with $N_h=\hat N_h$ outperforms that scheme with $N_h=\bar N_h$ in terms of the achievable secrecy rate.

\subsection{For the case of $N_s > N_{ep}$}
It holds that $\delta= \lfloor \dfrac{N_{\rm sum}-|N_s-N_{ep}|}{3}\rfloor+(N_s-N_{ep})$.

Let $\hat N_d=2\lfloor \dfrac{N_{\rm sum}-|N_s-N_{ep}|}{3}\rfloor+j+(N_s-N_{ep})$, and
\begin{align}
\hat N_h=\lfloor \dfrac{N_{\rm sum}-|N_s-N_{ep}|}{3}\rfloor, \label{eqa4}
\end{align}
where $j \triangleq N_{\rm sum}-3 \hat {N}_b^t$. By definition, $j \in \{0,1,2\}$.
Besides, since $\hat N_h < \hat N_d$, it holds that
\begin{align}
d_{c=1}(\hat N_h)=N_s-N_{ep}. \label{eqa20}
\end{align}\\
{\emph {B. 1 When $\hat N_h \ge N_{ep}$}}

In this subcase, it can be verified that $N_s \le \delta$ and $N_s \le N_{\rm sum}$. Thus, (\ref{eqa0}) becomes
\begin{align}
d_{s,p}=N_s.  \label{eqa51}
\end{align}
On the other hand, since $\hat N_h \ge N_{ep}$, it holds that
\begin{align}
s_2(\hat N_h)=N_{ep}. \label{eqa21}
\end{align}
Substituting (\ref{eqa20}) and (\ref{eqa21}) into (\ref{eqa1}) yields
$g(\hat N_h)= N_s$. In addition, by (\ref{eqa1}) the inequality $g( N_h) \le N_s$ always holds true.
Therefore, the maximum value of $g( N_h)$ over $N_h$ is
\begin{align}
g(\hat N_h)=N_s \mathop = \limits^{(a)} d_{s,p},  \nonumber
\end{align}
where (a) comes from the equality in (\ref{eqa51}).
%On the other hand, it holds that
%\begin{align}
%(N_s-N_{ep})+\lfloor \dfrac{N_{\rm sum}-|N_s-N_{ep}|}{3}\rfloor \ge N_s. \label{eqa11}
%\end{align}
%In addition, $N_{\rm sum}=3\lfloor \dfrac{N_{\rm sum}-|N_s-N_{ep}|}{3}\rfloor + j+(N_s-N_{ep})$.
%Thus,
%\begin{align}
%N_{\rm sum} > N_s, \nonumber
%\end{align}
%which, combined with (\ref{eqa11}), indicates that $d_{s,p}=N_s$.
%Therefore, the maximum value of $g( N_h)$ over $N_h$ is $g(\hat N_h)=d_{s,p}$.
\\\\
{\emph {B. 2 When $\hat N_h < N_{ep}$}}

In this subcase, it can be verified that $\delta \le N_s$ and $\delta \le N_{\rm sum}$. Thus, (\ref{eqa0}) becomes
\begin{align}
d_{s,p}=\delta=\hat N_h+(N_s-N_{ep}).  \label{eqa14}
\end{align}
On the other hand, $\hat N_h < N_{ep}$ combined with (\ref{eqa2b}), gives
\begin{align}
d_{c=2}^\star(\hat N_h)=\hat N_h. \label{eqa22}
\end{align}
Substituting (\ref{eqa20}) and (\ref{eqa22}) into (\ref{eqa1}) yields
\begin{align}
g(\hat N_h)= \hat N_h +(N_s-N_{ep}) \mathop = \limits^{(a)} d_{s,p}, \nonumber
\end{align}
where (a) comes from the equality in (\ref{eqa14}).

\emph{In the sequel, we will prove that for any other $N_h \ne \hat N_h$ it holds that
$g( N_h)\le d_{s,p}$, thus completing the proof of
that the maximum value of $g( N_h)$ over $N_h$ is $g(\hat N_h)=d_{s,p}$.}

i) For any $N_h < \hat N_h$, it holds that $d_{c=1}( N_h)=N_s- N_{ep}$ and
$d_{c=2}^\star(  N_h)=N_h < \hat N_h $.
Therefore,
\begin{align}
g(N_h)  \le  d_{c=1}( N_h) +d_{c=2}^\star(N_h) \le d_{s,p}. \label{eqa52}
\end{align}

ii) For any $N_h$ satisfying $N_h  > \hat N_h$ and $ N_h \le N_d$,
it holds that $d_{c=1}( N_h)=N_s- N_{ep}$. Based on (\ref{eqa2b}) it holds that
\begin{align}
&d_{c=2}^\star( N_h)  \le \lfloor (N_d-d_{c=1}( N_h))^+ / 2 \rfloor \nonumber \\
&\le \lfloor (\hat N_d-1-d_{c=1}( N_h))^+ / 2 \rfloor \nonumber \\
&=\hat N_h+\lfloor (j-1)/2 \rfloor, \nonumber
\end{align}
which, combined with the fact $j \le 2$, indicates that, $d_{c=2}^\star( N_h) \le \hat N_h$.
Therefore, the inequalities in (\ref{eqa52}) also hold true.

iii) For any $N_h$ satisfying $ N_h> \hat N_h$ and $ N_h> N_d$, we will
first give a specific value of $N_h$, denoted by $\bar N_h$, which satisfies $g(\bar N_h) \le d_{s,p}$.
We then prove that for any other $N_h \ne \bar N_h$ % satisfying $ N_h> \hat N_h$ and $ N_h> N_d$,
it holds that $g(N_h) \le g(\bar N_h) $. In this way, we finish the proof that $g(N_h) \le d_{s,p} $.

Note that since $N_{\rm sum} =N_h+ N_d > 2N_d$, for the case of $N_{\rm sum} \le 2(N_s-N_{ep})$ it holds that $N_d <(N_s-N_{ep})$, which, combined with
$g(N_h)\le N_d$, indicates that $g(N_h) < N_s-N_{ep} < d_{s,p}$.
Therefore, in the following arguments we only need to focus on the case of $N_{\rm sum} > 2(N_s-N_{ep})$.

Let $\bar N_d=\lfloor \dfrac{N_{\rm sum}-2|N_s-N_{ep}|}{3}\rfloor+(N_s-N_{ep})$, and
\begin{align}
\bar N_h=2\lfloor \dfrac{N_{\rm sum}-2|N_s-N_{ep}|}{3}\rfloor+\tau+(N_s-N_{ep}), \label{eqa18}
\end{align}
where $\tau \triangleq N_{\rm sum}-3\lfloor \dfrac{N_{\rm sum}-2(N_s-N_{ep})}{3}\rfloor-2(N_s-N_{ep})$. By definition, it holds that $\tau \in \{0, 1, 2\}$.
%\begin{align}
%\tau = N_{\rm sum}-3\lfloor \dfrac{N_{\rm sum}-2(N_s-N_{ep})}{3}\rfloor-2(N_s-N_{ep}). \label{eqa17}
%\end{align}
%Substituting (\ref{eqa17}) into $\bar N_h=N_{\rm sum}-\bar N_d$ yields

Substituting (\ref{eqa18}) into (\ref{eqa2a}), we arrive at
\begin{align}
\small d_{c=1}(\bar N_h)=N_s-N_{ep}+ \min\{\lfloor \dfrac{N_{\rm sum}-2|N_s-N_{ep}|}{3}\rfloor+\tau, N_{ep}\}, \nonumber %\min\{\bar N_h- \bar N_d,N_{ep}\}
\end{align}
which, combined with (\ref{eqa1}), gives
\begin{align}
\small g(\bar N_h)=\bar N_d=\lfloor \dfrac{N_{\rm sum}-2|N_s-N_{ep}|}{3}\rfloor+(N_s-N_{ep}) . \label{eqa15}
\end{align}
On comparing (\ref{eqa14}) and (\ref{eqa15}), one can see that
\begin{align}
g(\bar N_h) \le  d_{s,p}. \label{eqa24}
\end{align}

On the other hand, for any $N_h < \bar N_h$, say $N_h = \bar N_h-k$, $k \ge 1$, it holds that
$N_d = \bar N_d+k$. Thus, $N_h-N_d = \bar N_h- \bar N_d -2k <N_{ep}$,
which together with (\ref{eqa2a}), indicates that
\begin{align}
d_{c=1}(N_h)&=(N_s-N_{ep})+\lfloor \dfrac{N_{\rm sum}-2|N_s-N_{ep}|}{3}\rfloor+\tau-2k \nonumber\\
&\mathop = \limits^{(a)} g(\bar N_h) +\tau-2k, \nonumber
\end{align}
where (a) is due to (\ref{eqa15}).
In addition, by (\ref{eqa2b}) we have
\begin{align}
d_{c=2}^\star (N_h)   \le \lfloor (N_d-d_{c=1}(N_h))^+ / 2 \rfloor
\le \lfloor \dfrac{3k- \tau}{2}\rfloor. \nonumber
\end{align}
Since $\tau \le 2$ and $k \ge 1$,
it holds that $\tau-2k +\lfloor \dfrac{3k-\tau}{2}\rfloor \le 0$.
Therefore,
\begin{align}
g( N_h)\le d_{c=1}(N_h)+d_{c=2}^\star(N_h) \le g(\bar N_h). \label{eqa19}
\end{align}
Moreover, for any $N_h >\bar N_h$, it holds that
\begin{align}
g(N_h) \le  N_d <\bar N_d=g(\bar N_h). \label{eqa23}
\end{align}

Combining (\ref{eqa19}) with (\ref{eqa23}), one can see that for any other $N_h \ne \bar N_h$
satisfying $ N_h> \hat N_h$ and $ N_h> N_d$,
it holds that $g(N_h) \le g(\bar N_h) $, which, combined with (\ref{eqa24}),
indicates that $g(N_h) \le d_{s,p}$.
This completes the proof.

\section{Proof of \emph{Theorem 1}} \label{appC}
In the sequel, we will consider three distinct cases.
\begin{enumerate}
\item For the case of $N_e^t \ge N_e^r$, \emph{Eve} cannot see any interference-free subspaces,
and so the maximum achievable S.D.o.F. is equal to $\mathop{\lim }\limits_{ P \to \infty }
\dfrac{R_b}{{\rm log} \ P}$, whose maximum value over the input covariance matrices
is $\min\{(N_b-N_e^t)^+, N_a\}$. In that case, there is no need for \emph{Bob} to transmit
jamming signals to reduce the interference-free subspace that \emph{Eve} can see, and so
we set ${N_b^t}^\star=0$.
\item For the case of $N_e^t < N_e^r$ and $N_e^t \ge N_b$ the maximum achievable S.D.o.F. is
zero since \emph{Bob} already cannot see any interference-free subspaces. In that case, the achievable S.D.o.F.
will be zero even if \emph{Bob} transmits jamming signals, and so we set ${N_b^t}^\star=0$.
\item For the case of $N_e^t < N_e^r$ and $N_e^t < N_b$, no positive S.D.o.F. can be achieved if $N_b^r \le N_e^t$, and thus, in order
to maximize the achievable S.D.o.F., \emph{Bob} should choose a value
of $N_b^r$ such that $N_b^r >N_e^t $.
%Therefore, we only need to consider the case of $N_e^t < N_b^r$ and $N_b^r >N_e^t $;
In that case, and by \emph{Proposition 1}, one can see that the maximum achievable S.D.o.F.
is equal to that of a helper-assisted wiretap channel with number of antennas $N_s=N_a$, $ N_h=N_b^t$, $N_d=N_b^r-N_e^t$, $N_{\rm sum}=N_b-N_e^t$ and $N_{ep}=N_e^r-N_e^t$. Substituting these values
into \emph{Proposition 2}, we arrive at we arrive at the expression of ${N_b^t}^\star$, i.e., ${\hat N_h}$,
and also the maximum achievable S.D.o.F., i.e., $\min\{\eta, N_b-N_e^t, N_a\}$.
\end{enumerate}
%Based on the above, for the case of $N_e^t < N_e^r$ the maximum achievable S.D.o.F. equals $\min\{\eta, (N_b-N_e^t)^+, N_a\}$.

Concluding the above three cases, one can obtain the expressions of $d_{s,a}(N_e^t)$ and ${N_b^t}^\star$, as given in
(\ref{eq6}) and (\ref{eqNbt}), respectively.
This completes the proof.

\section{Proof of \emph{Theorem 2}} \label{appD}
We should note that for the case of $N_e \ge N_b$, the best choice for \emph{Eve} is to
allocate $N_b$ antennas to transmit; for that case no positive
S.D.o.F. can be achieved.
In what follows, we only need to study the nontrivial case of $N_e < N_b$.

From (\ref{eq6}), one can see that the achievable S.D.o.F. for the case of
$N_e^r < N_e^t$ is no greater than that of the other case. Therefore, to make sure that the achievable S.D.o.F.
is minimized, \emph{Eve} would always choose the value of $N_e^t$ such that $N_e^t < N_e^r$;
for that case
\begin{align}
d_{s,a}(N_e^t) =\min\{\eta, N_b-N_e^t, N_a\}, \label{eqc1}
\end{align}
with $\eta \triangleq \lfloor \frac{(N_b-N_e^t-|N_a-N_e^r+N_e^t|)^+}{3}\rfloor+(N_a-N_e^r+N_e^t)^+$.

Looking into the expression of $\eta$, we get two thresholds of $N_e^t$, i.e., $\dfrac{N_e-N_a}{2}$
and $ \dfrac{N_b+N_e-N_a}{3}$. Since $N_e < N_b$, it holds that $\dfrac{N_e-N_a}{2}< \dfrac{N_b+N_e-N_a}{3}$.
In order to simply the expression of $d_{s,a}(N_e^t)$,
in the following we will consider three distinct cases, which are obtained by those two thresholds.
\begin{enumerate}
\item For the case of $N_e^t \le \dfrac{N_e-N_a}{2}$, it holds that %$N_a \le N_e^r-N_e^t$, and so
\begin{align}
\eta &=\lfloor \dfrac{N_b+N_a-N_e+N_e^t}{3}\rfloor \le \lfloor \dfrac{N_b-N_e^t+N_a+N_e}{3}\rfloor \nonumber\\
& \mathop  \le \limits^{(a)} N_b-N_e^t, \nonumber
\end{align}
where (a) comes from the fact that
\begin{align}
N_a+N_e \le 2(N_e-N_e^t)<2(N_b-N_e^t). \nonumber
\end{align}
Thus, (\ref{eqc1}) becomes
\begin{align}
m_1(N_e^t)= \min\{\lfloor \dfrac{N_b+N_a-N_e+N_e^t}{3}\rfloor, N_a \}. \nonumber% \label{eqc1}
\end{align}
\item For the case of $ \dfrac{N_e-N_a}{2}< N_e^t < \dfrac{N_b+N_e-N_a}{3}$, it holds that
\begin{align}
\eta &=\lfloor \dfrac{N_b-N_a+N_e}{3}\rfloor+N_a-N_e+N_e^t. \nonumber
\end{align}
In addition, due to $N_e^t < \dfrac{N_b+N_e-N_a}{3}$ it holds that
\begin{align}
&2N_e^t \le 2\lfloor \dfrac{N_b+N_e-N_a}{3}\rfloor  \nonumber\\
&\Rightarrow 2N_e^t <  N_b+N_e-N_a-\lfloor \dfrac{N_b+N_e-N_a}{3}\rfloor  \nonumber \\
&\Rightarrow  \lfloor \dfrac{N_b+N_e-N_a}{3}\rfloor+N_a-N_e+N_e^t < N_b-N_e^t. \nonumber
\end{align}
Thus, (\ref{eqc1}) becomes
{\small{\begin{align}
m_2(N_e^t)=\min\{\lfloor \frac{N_b+N_e-N_a}{3}\rfloor+N_a-N_e+N_e^t, N_a \}. \nonumber%\label{eqc2}
\end{align}}}
\item For the case of $N_e^t \ge \dfrac{N_b+N_e-N_a}{3}$, it holds that
\begin{align}
\eta &=N_a-N_e+2N_e^t. \nonumber
\end{align}
Besides, it holds that $N_b-N_e^t \le N_a-N_e+2N_e^t$, which, combined with $2N_e^t < N_e$, indicates that $N_b-N_e^t < N_a$.
Thus, (\ref{eqc1}) becomes
\begin{align}
m_3(N_e^t)= N_b-N_e^t. \nonumber%\label{eqc3}
\end{align}
\end{enumerate}
Concluding the above three cases, one can see that
\begin{align}
d_{s,a}^{\rm{wc}} &= \mathop{\min }\limits_{ 0 \le N_e^t \le N_e } \min \{m_1(N_e^t), m_2(N_e^t), m_3(N_e^t)\}. \label{eqc8}
%&= \min  \{m_1(0), m_3(N_e),\mathop{\min }\limits_{  N_e^t} m_2(N_e^t)\}.  \label{eqc8}
\end{align}

In the sequel, we will consider three distinct cases, according to whether $m_i(N_e^t)$, $i=1,2,3$, is feasible.
For example, for the case of $N_e<N_a$, $m_1(N_e^t)$ is infeasible, since by definition it ranges
$N_e^t \le \dfrac{N_e-N_a}{2} <0$ which is unavailable.

\subsection{When $\max\{\dfrac{N_b-N_a}{2}, N_a\} \le N_e < N_b$}
It holds that
$\dfrac{N_e-N_a}{2} \ge 0$ and $\dfrac{N_b+N_e-N_a}{3} \le N_e$, which indicates that
both $m_1(N_e^t)$ and $m_3(N_e^t)$ are feasible. Moreover,
\begin{align}
&\mathop{\min }\limits_{N_e^t \le \frac{N_e-N_a}{2}} {m_1(N_e^t)}=m_1(0)=\min\{\lfloor \frac{N_b+N_a-N_e}{3}\rfloor, N_a \}, \nonumber\\
&\mathop{\min }\limits_{N_e^t \ge \frac{N_b+N_e-N_a}{3}} {m_3(N_e^t)}=m_3(N_e)=N_b-N_e. \nonumber
\end{align}

As to $m_2(N_e^t)$, it is feasible only for the case of
$\lfloor\dfrac{N_e-N_a}{2}\rfloor+1 < \dfrac{N_b+N_e-N_a}{3}$, in which
\begin{align}
&\mathop{\min }\limits_{\frac{N_e-N_a}{2} \le N_e^t \le \frac{N_b+N_e-N_a}{3}} m_2(N_e^t)
=m_2(\frac{N_e-N_a-\xi}{2}+1)\nonumber\\
&=\min\{\lfloor \frac{N_b+N_e-N_a}{3}\rfloor+\frac{N_a-N_e-\xi}{2}+1, N_a \}.\nonumber
\end{align}
Here, $\xi =1$ if $N_e-N_a$ is odd and otherwise $\xi =0$.

Since $N_a \le N_e<N_b$, it holds that
\begin{align}
\lfloor \frac{N_b-N_e+N_a}{3}\rfloor& \le \lfloor \frac{N_b+N_e-N_a}{3}\rfloor - \lfloor \frac{2(N_e-N_a)}{3}\rfloor. \nonumber
\end{align}
In addition, it can be verified that $\frac{N_e-N_a+\xi}{2}-1 \le \lfloor \frac{2(N_e-N_a)}{3}\rfloor$.
Therefore, we have $m_1(0) \le m_2(\frac{N_e-N_a-\xi}{2}+1)$.
%Combining (\ref{eqc4})-(\ref{eqc6}), we see that $m_1(0) \le m_2(\frac{N_e-N_a-\xi}{2}+1)$ when $N_e-N_a=1$.
%Besides, for the case of $N_e-N_a>1$ it can be verified that
%\begin{align}
%\lfloor \frac{2(N_e-N_a)}{3}\rfloor+1 > \frac{2(N_e-N_a)}{3} \ge \frac{N_e-N_a+\xi}{2}. \label{eqc7}
%\end{align}
%Combining (\ref{eqc4})-(\ref{eqc7}), we see that $m_1(0) \le m_2(\frac{N_e-N_a-\xi}{2}+1)$ also holds true.

Combining (\ref{eqc8}) with the above discussions, one can see that for the case
of $\max\{\dfrac{N_b-N_a}{2}, N_a\} \le N_e < N_b$,
\begin{align}
d_{s,a}^{\rm{wc}}&=\min \{m_1(0), m_3(N_e)\} \nonumber \\
&=\min\{\lfloor \frac{N_b+N_a-N_e}{3}\rfloor, N_b-N_e, N_a \}. \nonumber
\end{align}

\subsection{When $\dfrac{N_b-N_a}{2} \le N_e < \min\{ N_b, N_a\}$}
It holds that
$\dfrac{N_e-N_a}{2} < 0$ and $\dfrac{N_b+N_e-N_a}{3} \le  N_e$, which indicates that
$m_3(N_e^t)$ is feasible and $m_1(N_e^t)$ is infeasible. Moreover,
\begin{align}
&\mathop{\min }\limits_{N_e^t \ge \frac{N_b+N_e-N_a}{3}} {m_3(N_e^t)}=m_3(N_e)=N_b-N_e. \nonumber
\end{align}

$m_2(N_e^t)$ is feasible only for the case of $N_b-N_a+N_e > 0$,
in which case it holds that
\begin{align}
&\mathop{\min }\limits_{N_e^t \ge \frac{N_b+N_e-N_a}{3}} {m_2(N_e^t)}=m_2(0)\nonumber \\
&=\min\{\lfloor \frac{N_b+N_e-N_a}{3}\rfloor+N_a-N_e, N_b-N_e \}. \nonumber
\end{align}

Combining (\ref{eqc8}) with the above discussions, we have the following conclusions:

\begin{enumerate}
\item For the case of $\dfrac{N_b-N_a}{2} \le N_e < \min\{ N_b, N_a\}$ and $N_b-N_a+N_e > 0$,
it holds that
\begin{align}
&d_{s,a}^{\rm{wc}}=\min\{\lfloor \frac{N_b+N_e-N_a}{3}\rfloor+N_a-N_e, N_b-N_e \}. \nonumber
\end{align}
\item For the case of $\dfrac{N_b-N_a}{2} \le N_e < \min\{ N_b, N_a\}$ and $N_b-N_a+N_e \le 0$,
it holds that
\begin{align}
d_{s,a}^{\rm{wc}}=N_b-N_e. \nonumber
\end{align}
\end{enumerate}

\subsection{When $N_e < \min\{\dfrac{N_b-N_a}{2}, N_b \}$}
It holds that $\dfrac{N_b+N_e-N_a}{3} >  N_e$,
which indicates that $m_3(N_e^t)$ is infeasible, and $m_2(N_e^t)$ is feasible.

$m_1(N_e^t)$ is feasible only for the case of $N_e \ge N_a$, in which case it holds that
\begin{align}
d_{s,a}^{\rm{wc}}&= \min\{m_1(0), m_2(\frac{N_e-N_a-\xi}{2}+1) \} \nonumber \\
&\mathop  = \limits^{(a)}  m_1(0)\mathop  = \limits^{(b)} N_a.  \nonumber
\end{align}
where (a) is due to $m_1(0)\le m_2(\dfrac{N_e-N_a-\xi}{2}+1)$. (b)
is due to the fact that $\lfloor \dfrac{N_b+N_a-N_e}{3}\rfloor \ge N_a$, which is due to
$2N_e < {N_b-N_a}$ and $N_e \ge N_a$.

Also, for the case of $N_e < N_a$, we have
\begin{align}
&d_{s,a}^{\rm{wc}}=m_2(0)         \nonumber\\
&=\min\{\lfloor \frac{N_b+N_e-N_a}{3}\rfloor+N_a-N_e, N_a \} \nonumber \\
&=N_a. \nonumber
\end{align}

Concluding, for the case of $N_e < \min\{\dfrac{N_b-N_a}{2}, N_b \}$,
it holds that $d_{s,a}^{\rm{wc}}= N_a$. This completes the proof.

\bibliography{mybib}
\bibliographystyle{IEEEtran}

\end{document}